\documentclass{sig-alternate}

\usepackage[shortcuts]{extdash} 

\usepackage{amsmath,amssymb}
\usepackage{algorithm,algpseudocode}

\usepackage[nameinlink]{cleveref}
\crefname{algorithm}{Algorithm\@}{Algorithms\@}
\crefname{figure}{Figure\@}{Figures\@}
\crefname{table}{Table\@}{Tables\@}
\crefname{section}{Section\@}{Sections\@}
\crefname{lemma}{Lemma\@}{Lemmas\@}
\crefname{theorem}{Theorem\@}{Theorems\@}
\crefname{equation}{Equation\@}{Equations\@}
\crefname{line}{Line\@}{Lines\@}
\crefname{conjecture}{Conjecture\@}{Conjectures\@}
\crefname{appendix}{Section\@}{Sections\@}

\usepackage{color}

\newcommand{\lt}{\left}
\newcommand{\rt}{\right}
\newcommand{\M}[1]{\mathbf{\MakeUppercase{#1}}}

\newcommand\A{\M{A}}
\newcommand\B{\M{B}}
\newcommand\MC{\M{C}}
\newcommand\Q{\M{Q}}
\newcommand\R{\M{R}}
\newcommand\V{\M{V}}
\newcommand\T{\M{T}}

\newtheorem{theorem}{Theorem}
\newtheorem{lemma}{Lemma}

\newcommand\rQR{\textsc{rec\-/qr}}
\newcommand\QREG{\textsc{qr\-/eg}}
\newcommand\TSQR{\textsc{tsqr}}

\newcommand\CAQR{\textsc{caqr}}

\newcommand\ODQR{\textsc{\oldstylenums{1}d-caqr\-/eg}}
\newcommand\TDQR{\textsc{\oldstylenums{3}d-caqr\-/eg}}

\newcommand\ODMM{\textsc{\oldstylenums{1}dmm}}
\newcommand\TDMM{\textsc{\oldstylenums{3}dmm}}
\newcommand\MM{\textsc{mm}}

\newcommand\HQR{\textsc{\oldstylenums{1}d\-/house}}
\newcommand\HQRb{\textsc{\oldstylenums{2}d\-/house}}

\newcommand\Scatter{\textsc{scatter}}
\newcommand\Gather{\textsc{gather}}
\newcommand\Reduce{\textsc{reduce}}
\newcommand\Broadcast{\textsc{broadcast}}
\newcommand\AllGather{\textsc{all\-/gather}}
\newcommand\AllReduce{\textsc{all\-/reduce}}
\newcommand\AlltoAll{\textsc{all-to-all}}
\newcommand\ReduceScatter{\textsc{reduce\-/scatter}}

\newcommand\II{\mathcal{I}}
\newcommand\JJ{\mathcal{J}}
\newcommand\KK{\mathcal{K}}
\newcommand\cA{\mathcal{A}}
\newcommand\cB{\mathcal{B}}
\newcommand\cC{\mathcal{C}}

\title{A 3D Parallel Algorithm for QR Decomposition}

\doi{}
\isbn{}

\numberofauthors{5}
\author{%
\alignauthor Grey Ballard\\
\affaddr{Wake Forest University}\\
\affaddr{Winston Salem, NC, USA}\\
\email{ballard@wfu.edu}
\alignauthor James Demmel\\
\affaddr{University of California}\\
\affaddr{Berkeley, CA, USA}\\
\email{demmel@berkeley.edu}
\alignauthor Laura Grigori\\
\affaddr{INRIA Paris-Rocquencourt}\\
\affaddr{Paris, France}\\
\email{laura.grigori@inria.fr}
\and
\alignauthor Mathias Jacquelin\\
\affaddr{Lawrence Berkeley Natl.\ Lab.}\\
\affaddr{Berkeley, CA, USA}\\
\email{mjacquelin@lbl.gov}
\alignauthor Nicholas Knight\\
\affaddr{NYU--Courant}\\
\affaddr{New York, NY, USA}\\
\email{nknight@nyu.edu}
}
\date{}

\begin{document}

\maketitle
 
\begin{abstract}
Interprocessor communication often dominates the runtime of large matrix computations.
We present a parallel algorithm for computing QR decompositions whose bandwidth cost (communication volume) can be
decreased at the cost of increasing its latency cost (number of messages).
By varying a parameter to navigate the bandwidth\slash latency tradeoff, we can tune this algorithm for machines with different communication costs.
\end{abstract}

\section{Introduction}
\label{sec:introduction}

A common task in numerical linear algebra, especially when solving least\-/squares and eigenvalue problems, is 
\emph{QR\=/decomposing} a matrix into a unitary Q\=/factor times an upper trapezoidal R\=/factor.
We present a QR decomposition algorithm, \TDQR, 
whose bandwidth and latency costs demonstrate a tradeoff.

We model the cost of a parallel algorithm in terms of the number of arithmetic operations, the number of words moved between processors, and the number of messages in which these words are moved.
These three quantities, measured along critical paths in a parallel schedule, characterize the algorithm's arithmetic cost, bandwidth cost, and latency cost, respectively.
\begin{theorem}
\label{thm:3DUB}
An $m \times n$ matrix, $m \ge n$, can be QR\-/decomposed on $P$ processors with these asymptotic costs:
\begin{equation}
\label{eq:3DUB}
\begin{array}{c|c|c}
\text{\#\,operations} & \text{\#\,words} & \text{\#\,messages}
\\\hline
mn^2 / P & n^2 \slash \lt( nP/m\rt)^\delta & \lt( nP/m\rt)^\delta (\log P)^2
\end{array}
\end{equation}
where $\delta$ can be chosen from $[1/2,\,2/3]$, assuming
\begin{equation}
\label{eq:3DUB-hyp}
\begin{aligned}
P \mathbin{/} (\log P)^4 &= \Omega\big( m \mathbin{/} n \big)\text{,\quad and}
\\
P \cdot (\log P)^2 &= O\lt( m^{\frac{\delta}{1+\delta}} \cdot n^{\frac{1-\delta}{1+\delta}} \rt)\text.
\end{aligned}
\end{equation}
\end{theorem}

This arithmetic cost is optimal \cite{DGHL12}.
For the smallest $\delta = 1/2$, the latency cost is optimal, and
for the largest $\delta = 2/3$, the bandwidth cost is optimal 
\cite{BCDHKS14}.
However, these bandwidth and latency lower bounds are not attained simultaneously:
the bandwidth\-/latency product is $O(n^2 (\log P)^2)$.
We conjecture that this product must be $\Omega(n^2)$, meaning the tradeoff is inevitable.

Our main contribution is the presentation and analysis of \TDQR, which
extends Elmroth\-/Gustavson's recursive algorithm~\cite{EG00} to the distributed\-/memory setting and uses communication\-/efficient subroutines.
The inductive cases feature \emph{3D matrix multiplication} (\TDMM) \cite{ABGJP95}, which incurs a smaller bandwidth cost than conventional (2D) approaches.
The base cases feature a new variant of \emph{communication\-/avoiding QR} (\CAQR) \cite{DGHL12}.
\CAQR\ incurs a smaller latency cost than conventional (Householder) QR.
Our variant further improves the bandwidth cost.
We chose the name `\TDQR'\ to reflect this lineage.

For tall-and-skinny matrices whose aspect ratio is at least $P$,
it's best to directly invoke \TDQR's base-case subroutine, \ODQR.
\ODQR\ also demonstrates a bandwidth\slash latency tradeoff, albeit less drastic,
which we can navigate to derive the following bounds.
\begin{theorem}
\label{thm:1DUB}
An $m \times n$ matrix can be QR\-/decomposed on $P \ge m/n \ge 1$ processors with these asymptotic costs:
\begin{equation}
\label{eq:1DUB}
\begin{array}{c|c|c}
\text{\#\,operations} & \text{\#\,words} & \text{\#\,messages}
\\\hline
mn^2/P & n^2  & (\log P)^2
\end{array}
\end{equation}
assuming
$P\cdot(\log P)^2 = O(n^2)$.
\end{theorem}

The rest of this work is organized as follows.
We start by summarizing relevant mathematical background on computing QR decompositions (\cref{sec:QR}).
We then introduce our parallel machine model, formalizing how we quantify the costs of communication and computation (\cref{sec:model}).
Next we review the communication\-/efficient subroutines mentioned above, \TDMM\ (\cref{sec:3DMM}) and \CAQR\ (\cref{sec:CAQR}).
With this background in place, we present and analyze the new algorithms, \ODQR\ (\cref{sec:1DQR}) and \TDQR\ (\cref{sec:3DQR}), proving \cref{thm:3DUB,thm:1DUB}, in reverse order.
We conclude by discussing limitations and extensions and comparing with related work (\cref{sec:discussion}). 

\section{QR Decomposition}
\label{sec:QR}

In \cref{sec:QR}, we summarize the relevant background concerning computing QR decompositions.

After formalizing the problem in \cref{sec:QR:prelim}, we present a recursive template algorithm, called \rQR\ (\cref{alg:rQR} in \cref{sec:QR:recursive}), which includes many well\-/known algorithms as special cases.
We then specialize \rQR\ to utilize compact matrix representations (\cref{sec:QR:basis-kernel}) and a simpler recursive splitting strategy (\cref{sec:QR:QREG}).
The result of these specializations, called \QREG\ (\cref{alg:QREG}), serves as a template for our two new algorithms, \ODQR\ and \TDQR.

\subsection{QR Preliminaries}
\label{sec:QR:prelim}

A \emph{QR decomposition} of a matrix $\A$ is a matrix pair $(\Q,\R)$ such that 
$\A=\Q\R$, 
the \emph{Q\=/factor} $\Q$ is unitary, meaning $\Q^H\Q=\Q\Q^H=\M{I}$, and the \emph{R\=/factor} $\R$ is upper trapezoidal, meaning all entries below its main diagonal equal zero. 
To specialize for real\-/valued $\A$, simply substitute $(\cdot)^T$ for $(\cdot)^H$ and `orthogonal' for `unitary'. 

We will always assume that $\A$ has at least as many rows as columns.
(This implies that $\R$ has the same dimensions as $\A$ and is upper triangular.)
When $\A$ has more columns than rows, we can obtain a QR decomposition by splitting 
$\A = 
\begin{bmatrix}
\A_1 & \A_2 
\end{bmatrix}$
with square $\A_1$,
decomposing $\A_1 = \Q \R_1$, and 
computing
$\R = \begin{bmatrix} \R_1 & \Q^H \A_2 \end{bmatrix}$.

\subsection{Recursive QR Decomposition}
\label{sec:QR:recursive}

We consider QR decomposition algorithms based on \rQR\ (\cref{alg:rQR}), which split $\A$ vertically (\cref{line:rQR:split}), QR\=/decompose the left panel (\cref{line:rQR:left-rec}), update the right panel (\cref{line:rQR:update}), QR\=/decompose the lower part of the (updated) right panel (\cref{line:rQR:right-rec}), and then assemble a QR decomposition from the smaller ones (\cref{line:rQR:Q-assemble,line:rQR:R-assemble}).

\begin{algorithm}[htbp]
\caption{$\left(\Q,\,\R\right) = \rQR\left(\A\right)$}
\label{alg:rQR}
\begin{algorithmic}[1]

\If{\textsc{base\-/condition}}
	\label{line:rQR:base-condition}

	\State $\left(\Q,\,\R\right)=\textsc{base\=/QR}\left(\A\right)$.
	\label{line:rQR:base-case}
	
\Else

	\State \textsc{split} $\A = 
	\begin{bmatrix}
	\A_{11} & \A_{12} \\
	\A_{21} & \A_{22} 
	\end{bmatrix}$
	with $\A_{11}$ square.
	\label{line:rQR:split}

	\State $\left(\Q_\mathrm{L},\,\R_\mathrm{L}\right)=\rQR\left(\begin{bmatrix}\A_{11}\\\A_{21}\end{bmatrix}\right)$.
	\label{line:rQR:left-rec}
	
	\State $\begin{bmatrix}\M{B}_{12} \\ \M{B}_{22}\end{bmatrix} = \Q_\mathrm{L}^H\cdot\begin{bmatrix}\A_{12}\\\A_{22}\end{bmatrix}$.
	\label{line:rQR:update}
	
	\State $\left(\Q_\mathrm{R},\,\R_\mathrm{R}\right)=\rQR\left(\M{B}_{22}\right)$.
	\label{line:rQR:right-rec}	
		
	\State 
	$\Q = \Q_\mathrm{L} \cdot \begin{bmatrix} \M{I} & \M{0} \\ \M{0} & \Q_\mathrm{R} \end{bmatrix}$.
	\label{line:rQR:Q-assemble}
	
	\State
	$\R = \begin{bmatrix} \R_\mathrm{L} & \begin{bmatrix} \M{B}_{12} \\ \R_\mathrm{R} \end{bmatrix}\end{bmatrix}$.
	\label{line:rQR:R-assemble}

\EndIf

\end{algorithmic}
\end{algorithm}

We call \rQR\ a `template' because it leaves several details unspecified.
To instantiate this template and obtain an algorithm, we must pick
a base\-/case condition (\textsc{base-condition}, \cref{line:rQR:base-condition}), 
a base\-/case QR\=/decomposition subroutine (\textsc{base-QR}, \cref{line:rQR:base-case}), and
a splitting strategy (\textsc{Split}, \cref{line:rQR:split}).
Additionally, we must specify how the operations are scheduled and how the data are distributed.

\subsection{Compact Representations}
\label{sec:QR:basis-kernel}

In practice, a QR decomposition $(\Q,\R)$ of an $m\times n$ matrix ($m \ge n$) is typically not represented as a pair of explicit matrices.
In \cref{sec:QR:basis-kernel} we specialize \rQR\ to represent $\Q$ and $\R$ more compactly.

Since an R-factor is upper triangular, it is identifiable by just its superdiagonal entries.
Subsequently, the symbol $\R$ may either denote 
(1) the actual R-factor, an $m \times n$ upper\-/triangular matrix;
(2) the leading $n$ rows of the R-factor, an $n \times n$ upper\-/triangular matrix; or
(3) the upper triangle of the R-factor, a data structure of size $n(n+1)/2$.
When presenting algorithms, we will prefer convention~(2);
to obtain an $n \times n$ upper\-/triangular $\R$ from \rQR,
we agree that \textsc{base\=/QR} (\cref{line:rQR:base-case}) returns such an $\R$, and
we rewrite the R-factor assembly (\cref{line:rQR:R-assemble}) as
\[ \R = \begin{bmatrix} \R_\mathrm{L} & \M{B}_{12} \\ \M{0} & \R_\mathrm{R} \end{bmatrix}\text.\]

Any unitary matrix $\Q$ can be written as $\Q = \M{I}-\V\T\V^H$: the matrix pair $(\V,\T)$ is called a \emph{basis\-/kernel representation}~\cite{SB95} of $\Q$.
If $\Q$ is the Q-factor of a QR decomposition of an $m\times n$ matrix ($m \ge n$), then there exists such a representation where the \emph{basis} $\V$ is $m \times n$ and the \emph{kernel} $\T$ is $n\times n$.

Modifying \rQR\ to use basis\-/kernel representations, \cref{line:rQR:update} becomes
\begin{equation}
\label{eq:basis-kernel-update}
\begin{bmatrix}\M{B}_{12} \\ \M{B}_{22}\end{bmatrix} = 
\begin{bmatrix}\A_{12}\\\A_{22}\end{bmatrix} - \V_\mathrm{L}\T_\mathrm{L}^H\V_\mathrm{L}^H\begin{bmatrix}\A_{12} \\\A_{22}\end{bmatrix}\text,
\end{equation}
and \cref{line:rQR:Q-assemble} becomes 
\begin{equation}
\label{eq:basis-kernel-assemble} 
\begin{aligned}
\V&=\begin{bmatrix} \V_\mathrm{L} & \begin{bmatrix} \M{0} \\ \V_\mathrm{R} \end{bmatrix}\end{bmatrix}
\\
\T&= \begin{bmatrix} \T_\mathrm{L} & -\T_\mathrm{L}\V_\mathrm{L}^H\begin{bmatrix} \M{0} \\ \V_\mathrm{R} \end{bmatrix} \T_\mathrm{R} \\ \M{0} & \T_\mathrm{R} \end{bmatrix}\text,
\end{aligned}
\end{equation}
where $(\V_\mathrm{L},\T_\mathrm{L})$ and $(\V_\mathrm{R},\T_\mathrm{R})$ represent $\Q_\mathrm{L}$ and $\Q_\mathrm{R}$.

To simplify the presentation without affecting our asymptotic conclusions, we will not exploit the block\-/lower\-/trapezoidal and block\-/upper\-/triangular structures of the bases and kernels.
With this understanding, it minimizes arithmetic to evaluate the quadruple matrix product in \cref{eq:basis-kernel-update} from right to left, and the quadruple product in \cref{eq:basis-kernel-assemble} from inside\-/out (two possibilities).

When $m$ is close to $n$, a general basis\-/kernel representation may require more storage than the explicit ($m\times m$) Q-factor.
The QR decomposition algorithms in (Sca)LAPACK use a variant~\cite{P92} of \emph{compact WY representation}~\cite{SV89}, which we call \emph{Householder representation} in this work.
In Householder representation, $\V$ is unit lower trapezoidal and $\T$ is upper triangular.
These properties enable an in\-/place implementation, where $\V$'s strict lower trapezoid and $\R$'s upper triangle overwrite $\A$ and where $\T$ need not be stored, since in this case
\[ \T = \left(\mathrm{triu}(\V^H\V,{-}1)+\mathrm{diag}(\mathrm{diag}(\V^H\V))/2\right)^{-1} \text,\]
using the MATLAB operations `$\mathrm{triu}$' and `$\mathrm{diag}$'. 

Our algorithms will construct, store, and apply Q-factors in Householder representation.
This choice is motivated by our practical goal of integration into the ScaLAPACK library \cite{SCALAPACK}; 
from a theoretical standpoint, any basis\-/kernel representation (with $m \times n$ basis) would yield the same asymptotic costs. 

\subsection{Elmroth-Gustavson's Approach}
\label{sec:QR:QREG}

The recursive framework of \rQR\ is quite general.
We will obtain our desired algorithmic costs by following an approach of Elmroth\-/Gustavson~\cite{EG00} (implemented in LAPACK's \texttt{\_geqrt3}), in which we split $\A$ vertically (roughly) in half, until the number of columns drops below a given threshold $b \ge 1$.
The recursive calls define a binary tree whose $\lceil \log_2 (n/b) \rceil$ levels are complete except possibly the last.  
%
%
(We always suppose $b \le n$; when $b \ge n$, the tree has just one node.)
We call this specialized template \QREG\ (\cref{alg:QREG}); \QREG\ utilizes the compact representations as explained in \cref{sec:QR:basis-kernel}.

\begin{algorithm}
\caption{$\left(\V,\,\T,\,\R\right) = \QREG\left(\A,\,b\right)$}
\label{alg:QREG}
\begin{algorithmic}[1]

\If{$n \le b$}
	\label{line:QREG:base-condition}

	\State $\left(\V,\,\T,\,\R\right)=\textsc{base-QR}\left(\A\right)$.
	\label{line:QREG:base-case}
	
\Else

	\State \textsc{split} $\A = 
	\begin{bmatrix}
	\A_{11} & \A_{12} \\
	\A_{21} & \A_{22} 
	\end{bmatrix}$
	so $\A_{11}$ is $\lfloor n/2 \rfloor \times \lfloor n/2 \rfloor$.
	\label{line:QREG:split}

	\State $\left(\V_\mathrm{L},\,\T_\mathrm{L},\,\R_\mathrm{L}\right)=\QREG\left(\begin{bmatrix}\A_{11}\\\A_{21}\end{bmatrix},\,b\right)$.
	\label{line:QREG:left-rec}
	
	\State $\M{M}_1 = \V_\mathrm{L}^H \cdot\begin{bmatrix}\A_{12}\\\A_{22}\end{bmatrix}$.
	\label{line:QREG:update-1}
	
	\State $\M{M}_2 = \T_\mathrm{L}^H \cdot \M{M}_1$.
	\label{line:QREG:update-2}
	
	\State $\begin{bmatrix}\M{B}_{12} \\ \M{B}_{22}\end{bmatrix} = \begin{bmatrix}\A_{12}\\\A_{22}\end{bmatrix} - \V_\mathrm{L} \cdot \M{M}_2$. 
	\label{line:QREG:update-3}	
	
	\State $\left(\V_\mathrm{R},\,\T_\mathrm{R},\,\R_\mathrm{R}\right)=\QREG\left(\M{B}_{22},\, b\right)$.
	\label{line:QREG:right-rec}
		
	\State $\V = \begin{bmatrix} \V_\mathrm{L} & \begin{bmatrix} \M{0} \\ \V_\mathrm{R} \end{bmatrix}\end{bmatrix}$.	
	\label{line:QREG:V-assemble}
		
	\State $\M{M}_3 = \V_\mathrm{L}^H \cdot \begin{bmatrix} \M{0} \\ \V_\mathrm{R} \end{bmatrix}$.
	\label{line:QREG:T-assemble-1}	
	
	\State $\M{M}_4 = \M{M}_3 \cdot \T_\mathrm{R}$.
	\label{line:QREG:T-assemble-2}	
	
	\State $\T = \begin{bmatrix} \T_\mathrm{L} & -\T_\mathrm{L} \cdot \M{M}_4 \\ \M{0} & \T_\mathrm{R} \end{bmatrix}$.
	\label{line:QREG:T-assemble-3}	
	
	\State
	$\R = \begin{bmatrix} \R_\mathrm{L} & \M{B}_{12} \\ \M{0} & \R_\mathrm{R} \end{bmatrix}$.
	\label{line:QREG:R-assemble}

\EndIf

\end{algorithmic}
\end{algorithm}

We mention that \cite{EG00} actually proposes a hybrid of the stated approach and an iterative approach, switching between the two for a constant-factor improvement in the arithmetic cost.
(It still fits in the \rQR\ framework.) 
While our algorithms can also benefit from this optimization, we omit further discussion in the present work since it does not affect our asymptotic conclusions.

\section{Computation Model}
\label{sec:model}

We model a parallel machine as a set of $P$ interconnected processors, each with unbounded local memory.
Processors operate on local data and communicate with other processors by sending and receiving messages.
A processor can perform at most one task (operation\slash send\slash receive) at a time.
A (data) word means a complex number; operations are the usual field actions, plus complex conjugation and real square roots.
Messages are point\-/to\-/point and asynchronous.
Each operation takes time $\gamma$, while sending or receiving a message of $w$ words takes time $\alpha + w \beta$, $\alpha$ being the latency and $\beta$ the inverse of the bandwidth.

We model an execution as a DAG whose vertices are tasks and whose edges define $P$ paths, one for each processor's task sequence, plus an inter\-/path edge for each send\slash receive pair.   
Weighting vertices by their tasks' durations, we define runtime as the maximum weight of any path.
Therefore, if every path includes at most $F$ operations and at most $S$ messages, containing at most $W$ words in total, the runtime is bounded,
\[ \text{runtime} \le \gamma \cdot F + \beta \cdot W + \alpha \cdot S \text.\]

When multiple processors send\slash receive messages simultaneously, it can be more efficient to split the messages into more messages of smaller sizes, to coalesce them into fewer, larger messages, or to route them through intermediate processors.
These lower-level implementation details are often irrelevant to our asymptotic analyses, motivating us to express our algorithms' communication patterns abstractly, in terms of collectives.

In the rest of \cref{sec:model} we define the eight different collectives appearing in this work,
collecting in \cref{tab:collectives} their costs.

\begin{table*}\centering
\begin{tabular}{c|c|c|c}
     & \emph{\#\,operations}         & \emph{\#\,words}                           & \emph{\#\,messages} \\\hline\hline
\Scatter       & $0$                   & $(P-1)B$                           & $\log P$    \\\hline
\Gather        & $0$                   & $(P-1)B$                           & $\log P$    \\\hline
\Broadcast     & $0$                   & $\min(B\log P,\,B+P)$              & $\log P$    \\\hline
\Reduce        & $\min(B\log P,\,B+P)$ & $\min(B\log P,\,B+P)$              & $\log P$    \\\hline
\AllGather     & $0$                   & $(P-1)B$                           & $\log P$    \\\hline
\AllReduce     & $\min(B\log P,\,B+P)$ & $\min(B\log P,\,B+P)$              & $\log P$    \\\hline
\AlltoAll      & $0$                   & $\min(BP\log P,\,(B_*+P^2)\log P)$ & $\log P$    \\\hline
\ReduceScatter & $(P-1)B$              & $(P-1)B$                           & $\log P$
\end{tabular}
\caption{\label{tab:collectives}%
\textnormal{
Asymptotic costs of the collectives defined in \cref{sec:model}. 
$P$ is the number of processors involved,
$B = \max_{p,q} B_{pq}$ is the largest block\-/size, and
$B_* = \max\big(\max_q \sum_p B_{pq},\, \max_p \sum_q B_{pq}\big)$ is the maximum number of words any processor holds before/after.
Any $P$ can be replaced by $|\{p,q : B_{pq} > 0\}|$ for a possibly smaller (valid) bound.
}
}%
\end{table*}
 
A general scenario, called an \AlltoAll, is when every processor~$p$ initially owns a block of data, containing $B_{pq}$ words, destined for every processor~$q$, including itself.
All other collectives we use can be interpreted as special cases of an \AlltoAll.
Four of these distinguish a `root' processor $r$:
\Scatter, where only $r$'s outgoing blocks are nonempty;
\Gather, where only $r$'s incoming blocks are nonempty;
\Broadcast, a \Scatter\ where $r$'s outgoing blocks are identical; and
\Reduce, a \Gather\ where $r$'s incoming blocks have same size and are added entrywise.
Four others, including \AlltoAll, can be constructed from the first four:
\AllGather, $P$ \Gather{}s with same outgoing blocks but different roots;
\AllReduce, $P$ \Reduce{}s with same outgoing blocks but different roots;
\AlltoAll, $P$ \Gather{}s with different outgoing blocks and different roots; and
\ReduceScatter, $P$ \Reduce{}s with different outgoing blocks and different roots.
(Since the three `\Reduce' collectives perform arithmetic, they are technically not \AlltoAll{s}.)

\begin{lemma}
\label{lem:collectives}
There exist algorithms for the eight collectives satisfying the upper bounds in \cref{tab:collectives}.
\end{lemma}
\begin{proof}
For all but \AlltoAll\ we use a \emph{binomial\-/tree} and possibly a \emph{bidirectional\-/exchange algorithm}: see, e.g., \cite{TRG05,CHPV07}.
In particular, for (\textsc{all\=/})\allowbreak\Gather\ and (\textsc{reduce\=/})\allowbreak\Scatter\ we use binary tree algorithms,
and for \Broadcast\ and (\textsc{all\=/})\allowbreak\Reduce\ we use whichever of the two minimizes all three costs, asymptotically.
For \AlltoAll\ we use the \emph{(radix-2) index algorithm} \cite{BHKUW97}, 
possibly performed twice using the load\-/balancing approach of \cite{HBJ96}. 
(For a more detailed proof, see \cref{sec:collectives}.)
\end{proof}

Note that when the number of processors is not too large w.r.t.\ the block size, 
the bidirectional\-/exchange algorithm for \Broadcast, 
built from \Scatter\allowbreak+\AllGather,
is asymptotically cheaper than the corresponding binary tree algorithm in terms of bandwidth.
Similarly, the bidirectional\-/exchange algorithm for (\textsc{all\=/})\allowbreak\Reduce),
built from \ReduceScatter\allowbreak+(\textsc{all\=/})\allowbreak\Gather,
improves both arithmetic and bandwidth.

\section{3D Matrix Multiplication}
\label{sec:3DMM}

The key to reducing \TDQR's bandwidth cost below that of previous approaches is \emph{3D matrix multiplication} (\TDMM) \cite{ABGJP95}.
Here, \emph{3D} refers to the parallelization of the operations and distribution of data over a three-dimensional (logical) processor grid.
%
%
\ODQR\ will also exploit two special cases, \ODMM, performed on a one-dimensional grid, and \MM, performed locally on one processor.

For concreteness, consider multiplying an $I \times K$ matrix $\A$ with a $K \times J$ matrix $\B$ to obtain an $I \times J$ matrix $\MC$, via the usual (entrywise) formula: 
\begin{equation}
\label{eq:MMdef}
\text{for all $(i,j) \in [I]\times[J]$,}\quad
\MC_{ij} = \sum_{k \in [K]} \A_{ik} \B_{kj} \text.
\end{equation}
We identify each of the $IJK$ (scalar) multiplications with a point $(i,j,k)$ in three\-/dimensional Euclidean space, so the set of multiplications defines a discrete $I\times J \times K$ brick. 
We point the reader to \cite{BDKS16,SVP16} for further discussions of this geometric terminology.

\begin{lemma}
\label{lem:0DMM}
Suppose input matrices $\A$ and $\B$ are initially owned by processor $p$, and 
output matrix $\MC$ is to be finally owned also by processor $p$.
$\MC = \A \cdot \B$ can be computed with runtime
\begin{equation}
\label{eq:0DMM}
\gamma \cdot O\lt( I J K \rt) \text.
\end{equation}
\end{lemma}
\begin{proof}
Directly evaluating the sums\-/of\-/products in \cref{eq:MMdef} on processor $p$ involves $IJK$ multiplications and $IJ(K{-}1)$ additions; no communication is necessary.
\end{proof}

\begin{lemma}
\label{lem:1DMM}
Suppose $I$, $J$, $K$, and $P$ satisfy
\[ 
P = O\lt( \frac{IJK}{\max(I,J,K)} \rt)
\;\;\text{and}\;\; 
P = O\lt( \max(I,J,K) \rt)
\text.
\]
If $K = \max(I,J,K)$, suppose
that matrices $\A^T$ and $\B$ are initially distributed in matching row\-/wise layouts where each processor owns $O(K/P)$ rows, and 
that matrix $\MC$ is to be finally owned by a single processor $r$.
Alternatively, 
if $I = \max(I,J,K)$, suppose
that matrices $\A$ and $\MC$ are initially/finally distributed in matching row\-/wise layouts where each processor owns $O(I/P)$ rows, and
that matrix $\B$ is initially owned by a single processor $r$.
$\MC = \A \cdot \B$ can be computed with runtime
\begin{equation}
\label{eq:1DMM}
\gamma\cdot O\lt( \frac{IJK}{P} \rt) + \beta \cdot O\lt( \frac{IJK}{\max(I,J,K)} \rt) + \alpha \cdot O\lt(\log P\rt)\text, 
\end{equation}
\end{lemma} 
\begin{proof}
In the first case, each processor performs a local \MM\ and then all processors \Reduce\ to processor $r$.
In the second case, processor $r$ \Broadcast{s} $\B$ to all processors and then each processor performs a local \MM.
The hypotheses guarantee that $P$ is not too large for these collectives can leverage the bidirectional\-/exchange algorithms.
(For a more detailed proof, see \cref{sec:MM}.) 
\end{proof}
The bound of \cref{eq:1DMM} also holds in a third case, when $J = \max(I,J,K)$ and the distributions are symmetric to the second case, but we will not need this result.

\begin{lemma}
\label{lem:3DMM}
Suppose $I$, $J$, $K$, and $P$ satisfy
\[ c \cdot \frac{IJK}{\min(I,J,K)^3} \le P \le IJK\text. \]
There exists a data distribution of $\A$, $\B$, and $\MC$ such that each processor initially owns at most $O((I+J)K/P)$ entries of $\A$ and $\B$, and finally at most $O(IJ/P)$ entries of $\MC$, where
$\MC = \A \cdot \B$ can be computed with runtime
\begin{equation}
\label{eq:3DMM}
\gamma\cdot O\lt( \frac{IJK}{P} \rt) + \beta \cdot O\lt( \lt(\frac{IJK}{P}\rt)^{2/3} \rt) + \alpha \cdot O\lt(\log P\rt)\text.
\end{equation}
\end{lemma}
\begin{proof}
Pick 
$Q = \lfloor I / \rho \rfloor$, 
$R = \lfloor J / \rho \rfloor$, and
$S = \lfloor K / \rho \rfloor$, 
where
$\rho = (IJK/P)^{1/3}$.
%
%
Under the hypotheses, $Q,R,S$ are positive integers with $QRS \le P$, thus define a valid $Q \times R \times S$ processor grid. 
Moreover, $QRS = \Omega(P)$.
Pick partitions $\{\II_q\}_q$, $\{\JJ_r\}_r$, and $\{\KK_s\}_s$ of $[I]$, $[J]$, and $[K]$ which are \emph{balanced}, meaning their parts differ in size by at most one.
Pick a partition $\{\cA_{q,r,s}\}_{q,r,s}$ of $[I]\times[K]$ to be a union of balanced $R$-way partitions of the sets $\II_q \times \KK_s$ ($(q,s) \in [Q]\times[S]$), and similarly for partitions $\{\cB_{q,r,s}\}_{q,r,s}$ and $\{\cC_{q,r,s}\}_{q,r,s}$ of $[K]\times[J]$ and $[I]\times[J]$.
Distribute entries $\cA_{q,r,s}$ of $\A$ to each grid processor $(q,r,s)$, and similarly for $\B$ and $\MC$:
this distribution satisfies the balance constraint in the theorem statement.
The algorithm proceeds with \AllGather{s} of blocks of $\A$ and $\B$ along processor grid fibers in the $Q$- and $R$-directions. 
then local \MM{s}, 
then finally \ReduceScatter{s} of blocks of $\MC$ along processor grid fibers in the $S$-direction.
(For a more detailed proof, see \cref{sec:MM}.)
\end{proof}

We denote by \MM, \ODMM, or \TDMM\ an algorithm that satisfies \cref{lem:0DMM}, \cref{lem:1DMM}, or \cref{lem:3DMM}, resp.

\section{Communication-Avoiding QR}
\label{sec:CAQR}

Our new algorithms \ODQR\ and \TDQR\ are closely related to 
\emph{communication\-/avoiding QR (\CAQR)} 
and 
\emph{tall\-/skinny QR (\TSQR)} 
\cite{DGHL12}:
we explore this relationship in \cref{sec:discussion}.
For now we remark that 
\TDQR\ specializes \QREG\ to use \ODQR\ as a base-case,
and that
\ODQR\ specializes \QREG\ to use \TSQR\ (the variant in \cite{BDGJKN15}) as a base-case.

On input, the $m \times n$ matrix $\A$ is partitioned across the $P$ processors so that each processor~$p$ owns $m_p \ge n$ rows, not necessarily contiguous.
Thus we require $\A$ be sufficiently tall and skinny: $m/n \ge P$.
A single processor~$r$, which owns $\A$'s $n$ leading rows, is designated as the \emph{root processor}.

On output, the Q-factor is stored in Householder representation $(\V,\,\T)$, where $\V$ has the same distribution as $\A$.
Both $\T$ and the R-factor are returned only on the root processor.

\begin{lemma}
\label{lem:TSQR}
\textnormal{\TSQR}'s runtime is
\begin{multline*} 
\gamma\cdot O\lt( \max_p m_p n^2 + n^3 \log P\rt)
\\ + \beta\cdot O\lt( n^2 \log P\rt)
+ \alpha\cdot O\lt( \log P \rt)
\end{multline*}
\end{lemma}
\begin{proof}
It is crucial to use the \TSQR\ variant in \cite{BDGJKN15}.
(See also \cref{sec:TSQR} for a proof.)
\end{proof}

Recall from \cref{sec:model} that when the block-size is sufficiently large, \Reduce\ and \Broadcast\ can be performed more efficiently, by \ReduceScatter+\Gather\ and \Scatter+\AllGather, resp.
Unfortunately, \TSQR's \Reduce- and \Broadcast-like collectives preclude these optimizations.
Next in \cref{sec:1DQR}, we will show how similar savings are achievable.

\section{1D-CAQR-EG}
\label{sec:1DQR}

We now present a new algorithm, \ODQR, an instantiation of the template \QREG\ (\cref{alg:QREG}).
\ODQR\ effectively reduces \TSQR's bandwidth cost by a logarithmic factor, at the expense of increasing its latency cost by a comparable factor.

The input\slash output data distributions are the same as for \TSQR,
so we continue notation from \cref{sec:CAQR}.
We specify \ODQR\ by stepping line-by-line through \QREG\ ---
base case in \cref{sec:1DQR-bc} and inductive case in \cref{sec:1DQR-ic} ---
then prove \cref{thm:1DUB} in \cref{sec:1DQR-an}.

\subsection{Base Case}
\label{sec:1DQR-bc}

\ODQR's base-case QR decomposition subroutine (\cref{line:QREG:base-case}) is \TSQR\ (\cref{sec:CAQR}), using the same root processor.
Note that $\A$'s distribution satisfies \TSQR's requirements, and $\V$, $\T$, and $\R$ are returned distributed as required by \ODQR.
%

\subsection{Inductive Case}
\label{sec:1DQR-ic}

Let us walk through the inductive case line-by-line. 
All algorithmic costs are incurred in the two recursive calls (\cref{line:QREG:left-rec,line:QREG:right-rec}) and the six matrix multiplications (\cref{%
line:QREG:update-1,%
line:QREG:update-2,%
line:QREG:update-3,%
line:QREG:T-assemble-1,%
line:QREG:T-assemble-2,%
line:QREG:T-assemble-3}).

\vspace{1mm}\noindent\textbf{(\cref{line:QREG:split}):} 
the splitting involves no computation nor communication.

\vspace{1mm}\noindent\textbf{(\cref{line:QREG:left-rec}):}
the left recursive call is valid since 
$\left[\begin{smallmatrix}\A_{11}\\\A_{21}\end{smallmatrix}\right]$ 
still satisfies the data distribution requirements (only $n$ decreases).

\vspace{1mm}\noindent\textbf{(\cref{line:QREG:update-1}):} 
this is a \TDMM\ with matrix dimensions $I = \lfloor n/2 \rfloor$, $J = \lceil n/2 \rceil$, and $K = m$.
We choose a (1D) processor grid with $Q=R=1$ and $S=P$, thus $T=0$ and the partitions $\{\II_q\}_q = \{[I]\}$ and $\{\JJ_r\}_r = \{[J]\}$ are trivial.
We pick the partition $\{\KK_s\}_s$ to match the distribution of $\A$'s rows, and pick $\{\cA_{q,r,s}\}_{q,r,s} = \{ [I] \times \KK_s\}_s$ and $\{\cB_{q,r,s}\}_{q,r,s} = \{ \KK_s \times  [J]\}_s$.
Additionally, we set $\cC_{q,r,s}=\emptyset$ for all $(q,r,s)$ but the root processor.

\vspace{1mm}\noindent\textbf{(\cref{line:QREG:update-2}):} 
this is an \MM\ on the root processor with matrix dimensions $I=K=\lfloor n/2\rfloor$ and $J = \lceil n/2 \rceil$.

\vspace{1mm}\noindent\textbf{(\cref{line:QREG:update-3}):} 
this is a \TDMM\ with matrix dimensions $I = m$, $J = \lceil n/2 \rceil$, and $K=\lfloor n/2 \rfloor$, followed by a matrix subtraction.
We choose a (1D) processor grid with $Q=P$ and $R=S=1$, thus $T=0$ and the partitions $\{\JJ_r\}_r = \{[J]\}$ and $\{\KK_s\}_s = \{[K]\}$ are trivial.
We pick the partition $\{\II_r\}_r$ to match the distribution of $\A$'s rows, and pick $\{\cA_{q,r,s}\}_{q,r,s} = \{ \II_r \times [K]\}_r$ and $\{\cC_{q,r,s}\}_{q,r,s} = \{ \II_r \times [J]\}_r$.
Additionally, we set $\cB_{q,r,s}=\emptyset$ for all $(q,r,s)$ but the root processor.
The row-wise distribution $\{\cC_{q,r,s}\}_{q,r,s}$ enables the subsequent matrix subtraction to be performed without further communication.

\vspace{1mm}\noindent\textbf{(\cref{line:QREG:right-rec}):} 
the second recursive call is valid since $\M{B}_{22}$ still satisfies the data distribution requirements: 
the number of rows owned by the root processor decreases by the same amount that $n$ does, while
all other processors keep the same number.

\vspace{1mm}\noindent\textbf{(\cref{line:QREG:V-assemble}):} 
each processor assembles its local rows of $\V$: no computation nor communication is required.

\vspace{1mm}\noindent\textbf{(\cref{line:QREG:T-assemble-1}):}
this is a \TDMM\ with matrix dimensions $I = \lfloor n/2\rfloor$, $J = \lceil n/2 \rceil$, and $K = m{-}\lfloor n/2 \rfloor$;
we choose a processor grid and partitions as in \cref{line:QREG:update-1}.

\vspace{1mm}\noindent\textbf{(\cref{line:QREG:T-assemble-2}):} 
this is an \MM\ on the root processor with the same dimensions as \cref{line:QREG:update-2}.

\vspace{1mm}\noindent\textbf{(\cref{line:QREG:T-assemble-3}):} 
this is an \MM\ on the root processor with the same dimensions as \cref{line:QREG:update-2,line:QREG:T-assemble-2}.

\vspace{1mm}\noindent\textbf{(\cref{line:QREG:R-assemble}):} 
each processor assembles its local rows of $\R$: no computation nor communication is required.

\vspace{1mm}
Verify that $\V$, $\T$, and $\R$ are distributed as desired.

\subsection{Concluding the Analysis}
\label{sec:1DQR-an}

\ODQR\ is valid for any $P,m,n,b \ge 1$ such that $P \le m / n$, and there is no loss of generality to suppose $b \le n$.
When $b = n$, \ODQR\ reduces to \TSQR.
As we will see, picking $b < n$ allows us to reduce \ODQR's arithmetic and bandwidth costs --- while increasing its latency cost --- to appear as if we had used bidirectional exchange \Reduce\ and \Broadcast\ algorithms (\cref{sec:model}) within \TSQR, despite the fact that these algorithms are inapplicable, as we lamented at the end of \cref{sec:CAQR}.
We will navigate the tradeoff with a nonnegative parameter $\epsilon$, taking
\begin{align}
\label{eq:1DQR-b}
b &= \Theta\lt( n / (\log P)^\epsilon \rt)
\text.
\end{align}
We will show that taking $\epsilon = 1$ yields \cref{thm:1DUB}.

\begin{lemma}
\label{lem:1DQR}
If $P = O(b^2)$, \textnormal{\ODQR}\ has runtime
\begin{multline}
\label{eq:1DUB-alt}
\gamma \cdot \lt( \frac{mn^2}{P} + nb^2 \log P\rt) 
\\ + \beta\cdot O\lt( n^2 + nb\log P\rt) 
+\alpha\cdot O\lt( \frac{n}{b}\log P \rt)
\text.
\end{multline}
\end{lemma}
%
\begin{proof}
Here we give an expanded proof of \cref{lem:1DQR}.

Let us derive an upper bound $T(m,n)$ on the runtime of an \ODQR\ invocation.
(The unchanging parameters $P,b$ are implicit.)
We will now assume a balanced data distribution, meaning the numbers of rows any two processors owns differ by at most one.

When $P=1$, the runtime of \ODQR\ for any $b$ is just $\gamma \cdot O\lt( mn^2 \rt)$, which satisfies the conclusion, so we may assume $P > 1$ hereafter.

In the base case ($n \le b$), the algorithmic cost is the \TDQR\ call, so by \cref{lem:TSQR} we conclude that
\begin{multline*}
T(m,n) =
\gamma \cdot O\lt( \frac{mn^2}{P} + n^3 \log P \rt) 
\\+ \beta\cdot O\lt( n^2 \log P\rt) 
+\alpha\cdot O\lt(\log P \rt)
\text.
\end{multline*}

In the inductive case, the algorithmic cost is due to the two recursive calls, the three \ODMM\ calls (\cref{line:QREG:update-1,line:QREG:update-3,line:QREG:T-assemble-1}, performed on 1D processor grids, and the three (local) \MM\ calls (\cref{line:QREG:update-2,line:QREG:T-assemble-2,line:QREG:T-assemble-3}), performed by the root.

The local \MM{s} have runtime $\gamma \cdot O\lt( n^3 \rt)$.

To apply \cref{lem:1DMM} to the \TDMM\ calls, 
let us now suppose that $P = O\lt(n^2\rt)$ and $P = O(m)$, 
Actually, only the first assumption is new: we already know that $P \le m / n$ for the initial $m,n$, and when $P > 1$ we see that in any recursive call the current $m$ is within a factor of two of the initial $m$, hence $P=O(m)$ in any recursive call.
Confirming that the data distributions chosen in \cref{sec:1DQR-ic} match those in the proof of \cref{lem:1DMM}, the \ODMM{s}' runtime is
\[ 
\gamma \cdot O\lt( \frac{mn^2}{P} \rt) 
+ \beta\cdot O\lt( n^2 \rt) 
+\alpha\cdot O\lt(\log P \rt)
\text.
\] 

Overall, we have found that
\begin{multline*}
T(m,n) = 
T(m,\lfloor n/2) 
+ T(m {-} \lfloor n/2 \rfloor, \lceil n/2 \rceil) 
\\+ \gamma \cdot O\lt( \frac{mn^2}{P} \rt) 
+ \beta\cdot O\lt( n^2 \rt) 
+\alpha\cdot O\lt(\log P \rt)
\text.
\end{multline*}

By induction we may take $T(m,n)$ to be nondecreasing in both $m$ and $n$.
The former property justifies replacing $m{-}\lfloor n/2\rfloor$ by $m$ in the second recursive call.
Hence, we will take $m$ to be its initial value in the analysis of every recursive call.

Supposing $n = b2^L$ for a nonnegative integer $L$,
$T(m,n)$ is bounded by \cref{eq:1DUB-alt}.

Now observe that $n=b2^{L+1}$ reproduces the asymptotic bound \cref{eq:1DUB-alt}, and since $T(m,n)$ is nondecreasing in $n$, this bound also holds when $b2^L < n < b2^{L+1}$.

Requiring $P = O(b^2)$ suffices to ensure that $P = O(n^2)$ at every inductive case.
\end{proof}

\begin{proof}[of \cref{thm:1DUB}]
Substituting \cref{eq:1DQR-b} into \cref{eq:1DUB-alt},
\begin{multline*}
\gamma \cdot \lt( \frac{mn^2}{P}\lt( 1 + \frac{nP}{m}(\log P)^{1-2\epsilon} \rt) \rt) 
\\ + \beta\cdot O\lt( n^2 \lt(1 + (\log P)^{1-\epsilon}\rt) \rt) 
+\alpha\cdot O\lt( (\log P)^{1+\epsilon} \rt)
\text,
\end{multline*}
thence the hypothesis is $P(\log P)^{2\epsilon} = O\lt(n^2\rt)$.
We conclude by taking $\epsilon = 1$.
\end{proof}

This argument extends to any $\epsilon \ge 0$, assuming $P(\log P)^{2\epsilon} = O(n^2)$,
but the asymptotic tradeoff vanishes when $\epsilon > 1$.
For $1/2 \le \epsilon \le 1$ the tradeoff is only between bandwidth and latency.
A sensible interpretation of the case $\epsilon < 0$ is $b = n$, meaning \TSQR\ is invoked immediately.
In this case, the costs are given directly by \cref{lem:TSQR}.

\section{3D-CAQR-EG}
\label{sec:3DQR}

We now present our second new algorithm, \TDQR, another instantiation of the template \QREG\ (\cref{alg:QREG}).

On input, the $m\times n$ matrix $\A$ ($m \ge n$) is partitioned across the $P$ processors row\-/cyclically: thus, each processor owns at most $\lceil m / P \rceil$ rows.

On output, the Q-factor is stored in Householder representation $(\V,\,\T)$, where $\V$ has the same distribution as $\A$.
Both $\T$ and the R-factor have the same distribution, matching the top $n \times n$ submatrix of $\A$.

After walking through \TDQR\ in a similar fashion as we did for \ODQR\ in \cref{sec:1DQR}, we collect the results to prove \cref{thm:3DUB}.

In the following, $T_\TDQR$ denotes an upper bound on the runtime of $\TDQR$, a function of $P,m,n,b$.

\subsection{Base Case}
\label{sec:3DQR-bc}
Recall that $b$ denotes the recursive threshold for \TDQR.
\TDQR's base-case QR decomposition subroutine (\cref{line:QREG:base-case}) is \ODQR, with a fixed recursive threshold $b^*$.

To satisfy \ODQR's data distribution requirements, we convert $\A$ from row\-/cyclic to block\-/row layout, distributed over 
\[ P^* = \min\lt(P,\, \lfloor m/n \rfloor\rt) \]
processors, ensuring each owns at least $n$ rows and one owns the top $n$ rows (and perhaps others). 

Initially, $P' = \min(m, P)$ processors own rows of $\A$.
(Clearly $P' \le P$ with equality just in case $P \le m$;
note further that $P^* \le P'$ with equality just in case $P \le m/n$.) 
Number these processors from $0$ to $P'{-}1$ according to the cyclic layout of $\A$, so that processor~$0$ owns the top row of $A$.
Deal these processors among $P^*$ groups, so processor~$0$ goes into group~$0$, processor~$1$ goes into group~$1$, and so on. 
Represent each group by its lowest-numbered processor and within each group, \Gather\ $\A$'s rows to the representative.
Since each group contains at most $\lceil P' / P^* \rceil$ processors and each processor initially owns at most $\lceil m/P' \rceil$ rows of $\A$, the largest block-size in any \Gather\ is at most $\lceil m/P' \rceil n$.

Each of the $P^*$ representatives (including processor $0$) now owns at least $\lfloor m / P^* \rfloor \ge n$ rows of $\A$, satisfying the first part of \ODQR's data distribution requirements: it remains to ensure processor~$0$ owns the top $n$ rows of $\A$.
These rows are currently owned by the first $P''=\min(P^*,n)$ representatives.
(Clearly $P'' \le P^*$ with equality just in case $P \le n$.)
We next perform a \Gather\ over the representatives of groups~$0$ through~$P''{-}1$, taking processor~$0$ to be the root so that afterwards it owns the top $n$ rows of $\A$ (and perhaps others).
We also perform a \Scatter\ with the opposite communication pattern so that the overall number of rows per representative is unchanged.
The largest block-size in both the \Gather\ and the \Scatter\ is at most $\lceil n / P'' \rceil n$.

We can now invoke \ODQR, with parameters $P^*,b^*$.
After it returns, we redistribute $\V$, $\T$, and $\R$ by reversing the preceding \Gather{s}\slash\Scatter{s}, so that $\V$ is (resp., $\T$ and $\R$ are) distributed over all $P$ processors like $\A$ was (resp., $\A$'s first $n$ rows were) initially.

\subsection{Inductive Case}
\label{sec:3DQR-ic}

Let us walk through the inductive case line-by-line, as we did for \ODQR. 
All algorithmic costs are incurred in the two recursive calls (\cref{line:QREG:left-rec,line:QREG:right-rec}) and the six matrix multiplications (\cref{%
line:QREG:update-1,%
line:QREG:update-2,%
line:QREG:update-3,%
line:QREG:T-assemble-1,%
line:QREG:T-assemble-2,%
line:QREG:T-assemble-3}).

\vspace{1mm}\noindent\textbf{(\cref{line:QREG:split}):} 
the splitting involves no computation nor communication.

\vspace{1mm}\noindent\textbf{(\cref{line:QREG:left-rec}):}
the left recursive call is valid since 
$\left[\begin{smallmatrix}\A_{11}\\\A_{21}\end{smallmatrix}\right]$ 
still satisfies the data distribution requirements (only $n$ decreases).

\vspace{1mm}\noindent\textbf{(\cref{line:QREG:update-1}):} 
this is a \TDMM\ with matrix dimensions $I = \lfloor n/2 \rfloor$, $J = \lceil n/2 \rceil$, and $K = m$.
We do not yet specify the processor grid, but we do suppose that \TDQR\ uses a balanced parallelization and data distribution as in the proof of \cref{lem:3DMM}: this is possible for any processor grid. 

To match this data distribution, we perform an \AlltoAll\ before and after the \TDMM\ invocation, each time using the two-phase approach \cite{BHKUW97}.
The first \AlltoAll\ redistributes the input matrices from column- and row-cyclic to \TDMM\ layout (the left factor is row-cyclic, transposed);
the maximum number of input matrix entries any processor owns before or after this collective is at most
\[ \max\lt( 
I\lt\lceil\frac{K}{P}\rt\rceil + \lt\lceil\frac{K}{P}\rt\rceil J,\, 
\lt\lceil\frac{\lt\lceil\frac{I}{Q}\rt\rceil\lt\lceil\frac{K}{S}\rt\rceil}{R}\rt\rceil +
\lt\lceil\frac{\lt\lceil\frac{J}{R}\rt\rceil\lt\lceil\frac{K}{S}\rt\rceil}{Q}\rt\rceil 
\rt) \text, \]
where the processor grid is $Q \times R \times S$.
The second \AlltoAll\ converts the output matrix from \TDMM\ layout to row-cyclic layout;
the maximum number of output matrix entries any processor owns before or after this collective is at most 
\[ \max\lt( 
\lt\lceil\frac{\lt\lceil\frac{I}{Q}\rt\rceil\lt\lceil\frac{J}{R}\rt\rceil}{S}\rt\rceil,\,
\lt\lceil\frac{I}{P}\rt\rceil J 
\rt)\text. \] 

\vspace{1mm}\noindent\textbf{(\cref{line:QREG:update-2}):} 
this is a \TDMM\ with matrix dimensions $I=K=\lfloor n/2\rfloor$ and $J = \lceil n/2 \rceil$.
We pick a 3D processor grid and partitions satisfying the same constraints as for \cref{line:QREG:update-1}, and perform similar \AlltoAll{s} before and after, so we can reuse the preceding analysis, substituting $I$, $J$, and $K$.

\vspace{1mm}\noindent\textbf{(\cref{line:QREG:update-3}):} 
this is a \TDMM\ with matrix dimensions $I = m$, $J = \lceil n/2 \rceil$, and $K=\lfloor n/2 \rfloor$, followed by a matrix subtraction.
We proceed similarly to \cref{line:QREG:update-1,line:QREG:update-2}, 
except that the left factor is initially in row-cyclic layout, so the first summand in the first term of the maximum becomes $\lceil I/P\rceil K$ (vs.\ $I\lceil K/P\rceil$).

\vspace{1mm}\noindent\textbf{(\cref{line:QREG:right-rec}):} 
the second recursive call is valid since $\M{B}_{22}$ still satisfies the data distribution requirements:
in particular, unlike \ODQR\ there is no requirement that a fixed processor owns the first $n$ rows or that every processor owns at least $n$ rows.

\vspace{1mm}\noindent\textbf{(\cref{line:QREG:V-assemble}):} 
each processor assembles its local rows of $\V$: no computation nor communication is required.

\vspace{1mm}\noindent\textbf{(\cref{line:QREG:T-assemble-1}):}
this is a \TDMM\ with matrix dimensions $I = \lfloor n/2\rfloor$, $J = \lceil n/2 \rceil$, and $K = m{-}\lfloor n/2 \rfloor$;
we choose a processor grid, partitions, and \AlltoAll{s} as in \cref{line:QREG:update-1}.

\vspace{1mm}\noindent\textbf{(\cref{line:QREG:T-assemble-2}):} 
this is a \TDMM\ with the same dimensions, processor grid, partitions, and \AlltoAll{s} as \cref{line:QREG:update-2}.

\vspace{1mm}\noindent\textbf{(\cref{line:QREG:T-assemble-3}):} 
this is a \TDMM\ with the same dimensions, processor grid, partitions, and \AlltoAll{s} as \cref{line:QREG:update-2,line:QREG:T-assemble-2}. 

\vspace{1mm}\noindent\textbf{(\cref{line:QREG:R-assemble}):} 
each processor assembles its local rows of $\R$: no computation nor communication is required.

\vspace{1mm}
Verify that $\V$, $\T$, and $\R$ are distributed as desired.

\subsection{Concluding the Analysis}
\label{sec:3DQR-an}

\TDQR\ is valid for any $P,m,n,b,b^* \ge 1$, and there is no loss of generality to suppose $b^* \le b \le n$.
Taking $b = n$ simplifies \TDQR\ to \ODQR\ with parameters $P^*,b^*$ and additional data redistributions. 
As in the case of \ODQR, picking $b < n$ allows us to reduce \TDQR's arithmetic and bandwidth costs, while increasing its latency cost.

We will navigate this tradeoff with two nonnegative parameters $\delta,\epsilon$, taking
\begin{align}
\label{eq:3DQR-b}
b &= \Theta\lt( n / \lt(nP / m\rt)^\delta \rt)
\text,&
b^* &= \Theta\lt( b / (\log P)^\epsilon  \rt)
\text.
\end{align}
We prove \cref{thm:3DUB} with $\delta \in [1/2,2/3]$ and $\epsilon = 1$.

\begin{lemma}
\label{lem:3DQR}
If $P = O(b^2)$ and $P^* = O({b^*}^2)$, \textnormal{\TDQR}\ has runtime
\begin{multline}
\label{eq:3DUB-alt}
T(m,n) = 
\gamma \cdot \lt( \frac{mn^2}{P} + n{b^*}^2 \log P\rt) 
\\ + \beta\cdot O\Bigg(\frac{mn}{P} + nb + nb^* \log P  + \lt(\frac{mn^2}{P}\rt)^{2/3}
\\+ \lt(\lt(\frac{mn}{P} + n\rt)\log\frac{n}{b} + \frac{nP^2}{b}\rt) \log P\Bigg) 
\\+\alpha\cdot O\lt( \frac{n}{b^*}\log P \rt)
\text.
\end{multline}
\end{lemma}
%
\begin{proof}
Here we give an expanded proof of \cref{lem:3DQR}.

Let us derive an upper bound $T(m,n)$ on the runtime of a \TDQR\ invocation.
(The unchanging parameters $P,b,b^*$ are implicit.)

When $P=1$, the runtime of \TDQR\ is just $\gamma \cdot O\lt( mn^2 \rt)$, which satisfies the conclusion, so we may assume $P > 1$ hereafter.

In a base case ($n \le b$), the algorithmic costs are due to the \ODQR\ invocation and the four communication phases.

The first communication phase, involving $P^*$ independent \Gather{s}, has runtime bounded by
\begin{multline*}
\beta \cdot O\lt( (\lceil P'/P^*\rceil - 1)\lceil m/P' \rceil n \rt) + \alpha \cdot O\lt( \log \lceil P'/P^* \rceil \rt)
\\= \beta \cdot O\lt( mn/P + n^2 \rt) + \alpha \cdot O\lt( \log P \rt)\text,
\end{multline*}
and the same bound applies for the last phase (matching \Scatter{s}).
(Recall that $P^*=\min(P,\lfloor m/n\rfloor)$ and $P' = \min(P,m)$.)
The fact that no communication happens when $m \ge nP$ 
(and thus $P^*=P$)
is evident in the first bound but not the second.

The second communication phase, a simultaneous \Gather\slash\Scatter, has runtime bounded by
\begin{multline*}
\beta \cdot O(\lt( (P''-1)\lceil n/P''\rceil n\rt) + \alpha \cdot O\lt( \log P'' \rt)
\\= \beta \cdot O\lt( n^2 \rt) + \alpha \cdot O\lt( \log P \rt)\text,
\end{multline*}
and the same bound applies for the third phase (matching \Scatter\slash\Gather).
(Recall that $P'' = \min(P^*,n)$.)

A runtime bound for the \ODQR\ invocation is given by \cref{thm:1DUB}, supposing now that $P^* = O({b^*}^2)$.
Actually we use the more refined bound of \cref{eq:1DUB-alt}, substituting $P^*,b^*$ for $P,b$.

Altogether, in a base case, 
\begin{multline*}
T(m,n) = 
\gamma\cdot O\lt( \frac{mn^2}{P} + n{b^*}^2 \log P \rt) 
\\+\beta \cdot O\lt( \frac{mn}{P} + n^2 + nb^* \log P \rt)
+\alpha \cdot O\lt( \frac{n}{b^*} \log P \rt)
\end{multline*}

In the inductive case ($n > b$), the algorithmic cost is due to 
the two recursive calls, 
the six \TDMM{s}, performed on 3D processor grids, and 
the twelve \AlltoAll{s}, performed before and after each \TDMM.

To apply \cref{lem:3DMM} to the \TDMM{s}, 
let us now suppose that $P \ge (3c)^{3/4}$ for some $c > 1$ and $b \ge 2 P^{1/3}$:
since $m \ge n$ at every recursive call and since $n \ge b+1$ in the inductive case,
$c\cdot IJK/\min(I,J,K)^3 \le P \le IJK$ for each \TDMM.
%
%
Thus the \TDMM{s}' overall runtime is
\[ 
\gamma\cdot O\lt( \frac{mn^2}{P} \rt) + \beta \cdot O\lt( \lt(\frac{mn^2}{P}\rt)^{2/3} \rt) + \alpha \cdot O\lt(\log P\rt)
\text.\]
Moreover, the inequalities derived at the beginning of \cref{lem:3DMM}'s proof also yield the following upper bound on the overall runtime of the \AlltoAll{s},
\[ \beta \cdot O\lt(\lt(\frac{mn}{P} + n + P^2\rt)\log P\rt) + \alpha \cdot O\lt( \log P \rt)\text.\]

Overall, we have found that
\begin{multline*}
T(m,n) = 
T(m,\lfloor n/2\rfloor) 
+ T(m {-} \lfloor n/2 \rfloor, \lceil n/2 \rceil) 
\\+ \gamma \cdot O\lt( \frac{mn^2}{P} \rt) 
+\beta\cdot O\Bigg( \lt(\frac{mn^2}{P}\rt)^{2/3} 
\\+ \lt(\frac{mn}{P} + n + P^2\rt)\log P\Bigg)
+\alpha\cdot O\lt(\log P \rt)
\text.
\end{multline*}

By induction we may take $T(m,n)$ to be nondecreasing in both $m$ and $n$.
The former property justifies replacing $m{-}\lfloor n/2\rfloor$ by $m$ in the second recursive call.
Hence, we will take $m$ to be its initial value in the analysis of every recursive call.

Supposing $n = b2^L$ for a nonnegative integer $L$,
$T(m,n)$ is bounded by \cref{eq:3DUB-alt}.

Now observe that $n=b2^{L+1}$ reproduces the asymptotic bound of \cref{eq:3DUB-alt}, 
which thus holds when $b2^L < n < b2^{L+1}$ since $T(m,n)$ is nondecreasing in $n$.

In conclusion, 
\TDQR's runtime $T(m,n)$ satisfies \cref{eq:3DUB-alt} if $P \ge 3$, $b \ge 2P^{1/3}$, $P = O(b^2)$, and $P^* = O({b^*}^2)$.
\end{proof}

\begin{proof}[of \cref{thm:3DUB}]
Consider picking $b,b^*$ as in \cref{eq:3DQR-b}.
The constraints relating $P$ and $b$ are satisfiable if
\begin{equation}
\label{eq:3DQR-bconst}
P = O\lt( m^{\frac{2\delta}{1+2\delta}} \cdot n^{\frac{2-2\delta}{1+2\delta}} \rt)\text,
\end{equation}
and the constraint relating $P^*$ and $b^*$ is satisfiable if
\begin{equation}
\label{eq:3DQR-bstarconst}
P \cdot (\log P)^{\frac{2\epsilon}{1+2\delta}} = O\lt( m^{\frac{2\delta}{1+2\delta}} \cdot n^{\frac{2-2\delta}{1+2\delta}} \rt)\text,
\end{equation}
a stronger condition than \cref{eq:3DQR-bconst}.

Substituting \cref{eq:1DQR-b} into \cref{eq:1DUB-alt},
\begin{multline*}
\gamma \cdot \lt( \frac{mn^2}{P}\lt( 1 + \lt(\frac{nP}{m}\rt)^{1-2\delta}(\log P)^{1-2\epsilon} \rt) \rt) 
\\ + \beta\cdot O\lt( \frac{n^2}{\lt(\frac{nP}{m}\rt)^\delta} \lt(1 + (\log P)^{1-\epsilon}\rt) + W \rt) 
\\+\alpha\cdot O\lt( (nP/m)^\delta (\log P)^{1+\epsilon} \rt)
\text,
\end{multline*}
where $W$ denotes the sum of three terms associated with the \AlltoAll{s}, 
\[ 
\frac{mn}{P}\log\frac{nP}{m} \log P
\text,\;\;
n \log\frac{nP}{m} \log P
\text,\;\;
P^2\lt(\frac{nP}{m}\rt)^\delta \log P
\text,
\]
plus a term $n^2 / (nP/m)^{2/3}$ associated with the \TDMM{s}.
We obtain the stated arithmetic and latency costs by taking $\delta \ge 1/2$ and $\epsilon = 1$.
To suppress the bandwidth term $W$, it suffices to require that there exists $\delta' \in (0,1{-}\delta)$, hence $\delta < 1$, such that
\begin{equation}
\label{eq:3DQR-aconst}
\begin{aligned}
P \mathbin{/} (\log P)^{\frac{\epsilon}{1-\delta-\delta'}} &= \Omega \big(m \mathbin{/} n \big)\text,
\\
P \cdot (\log P)^{\frac{\epsilon}{\delta + \delta'}} &= O\lt( m \cdot n^{\frac{1-\delta-\delta'}{\delta+\delta'}} \rt)\text,
\\
P \cdot (\log P)^{\frac{\epsilon}{2+2\delta}} &= O\lt( m^{\frac{\delta}{1+\delta}} \cdot n^{\frac{1-\delta}{1+\delta}} \rt)\text.
\end{aligned}
\end{equation} 
The \TDMM{s}' bandwidth cost cannot be reduced, but it is lower\-/ order if $\delta \le 2/3$.

The hypotheses of \cref{thm:3DUB}, $\delta \in [1/2,2/3]$ (and tacitly $\epsilon = 1$) and \cref{eq:3DUB-hyp},
imply \cref{eq:3DQR-bconst,eq:3DQR-bstarconst,eq:3DQR-aconst}.  
\end{proof}

This argument extends to a larger range of nonnegative $\delta,\epsilon$.
Assuming fixed $\epsilon$,
for $\delta > 2/3$ the \TDMM\ invocations dominate the bandwidth cost, whose bound remains as if $\delta = 2/3$, no longer a tradeoff, at least asymptotically.
In the case $0 \le \delta < 1/2$, the additive term in the arithmetic cost,
due to the small (mostly) triangular matrix operations on \TSQR's critical path,
possibly dominates.
(A sensible interpretation of the case $\delta \le 0$ is $b = n$, in which case \ODQR\ is invoked immediately.)
Assuming fixed $\delta$, the tradeoffs due to varying $\epsilon \in [0,1]$ are just as in the proof of \cref{thm:1DUB},
except now the factor in the arithmetic cost is suppressed by increasing $\delta$.

\section{Discussion}
\label{sec:discussion}

We have presented two new algorithms, \ODQR\ and \TDQR\ (\cref{sec:1DQR,sec:3DQR}), for computing QR decompositions on distributed\-/memory parallel machines.
Our analysis (e.g., \cref{eq:1DUB-alt,eq:3DUB-alt}) demonstrates tradeoffs between arithmetic, bandwidth, and latency costs,
governed by the choice of one (\ODQR)\ or two (\TDQR)\ block sizes.
We navigated these tradeoffs in \cref{thm:1DUB,thm:3DUB} by asymptotically minimizing arithmetic, 
as well as bandwidth in \cref{thm:1DUB}.

\subsection{Comparison With Similar Algorithms}
\label{sec:compare}

Here we compare the two new algorithms with four other instances of the \rQR\ framework,
deriving \cref{tab:costcomp1D,tab:costcomp2D}. 
Let us review the other algorithms.

\begin{table*}\centering
\begin{tabular}{c|c|c|c}
\emph{algorithm} & \emph{\#\,operations} & \emph{\#\,words} & \emph{\#\,messages} \\\hline\hline
\HQRb & $mn^2 / P$ & $n^2 / (nP / m)^{1/2}$  & $n \log P$ \\
\CAQR & $mn^2 / P$ & $n^2 / (nP / m)^{1/2}$  & $(nP/m)^{1/2} (\log P)^2$\\
\TDQR & $mn^2 / P$ & $n^2 / (nP / m)^\delta$ & $(nP/m)^\delta (\log P)^2$ 
\end{tabular}
\caption{\label{tab:costcomp2D}
\textnormal{
Comparison of approaches for square-ish matrices ($m / n = O(P)$). 
The algorithms and the assumptions that support these bounds are explained in \cref{sec:compare}. 
(In line 3, $\delta$ varies from $1/2$ to $2/3$.)
}
}
\end{table*}

\begin{table*}\centering
\begin{tabular}{c|c|c|c}
\emph{algorithm} & \emph{\#\,operations} & \emph{\#\,words} & \emph{\#\,messages} \\\hline\hline
\HQR & $mn^2 / P$ & $n^2 \log P$  & $n \log P$ \\
\TSQR & $mn^2 / P + n^3 \log P$ & $n^2 \log P$  & $\log P$\\
\ODQR & $mn^2 / P + n^3 (\log P)^{1-2\epsilon}$ & $n^2 (\log P)^{1-\epsilon}$ & $(\log P)^{1+\epsilon}$ 
\end{tabular}
\caption{\label{tab:costcomp1D}
\textnormal{
Comparison of approaches for tall/skinny matrices ($m / n = \Omega(P)$). 
The algorithms and the assumptions that support these bounds are explained in \cref{sec:compare}. 
(In line 3, $\epsilon$ varies from $0$ to $1$.)
}
}
\end{table*}

An early and well-known instance of \rQR\ (\cref{alg:rQR}) was proposed by Householder~\cite{H58}.
It features a splitting strategy (\cref{line:rQR:split}) 
where $\A_{11}$ is $b \times b$ (\emph{right\-/looking}) or $n{-}b \times n{-}b$ (\emph{left\-/looking});
a base case threshold (\cref{line:rQR:base-condition}) asserting $n \le b$; and
 a base case subroutine (\cref{line:rQR:base-case}) generating a product of $b$ Householder reflectors. 
Householder's proposal was right\-/looking and \emph{unblocked}, meaning $b=1$ (vs.\ \emph{blocked}, $b>1$).

Let \HQR\ and \HQRb\ denote the un/blocked right\-/looking variants, specialized to use compact representations (\cref{sec:QR:basis-kernel});
their costs are summarized in the first rows of \cref{tab:costcomp1D,tab:costcomp2D}.
\HQRb\ invokes \HQR\ as its base case.
For \HQR\ we use a 1D processor grid and for \HQRb\ we use a 2D processor grid.
For \HQR\ we distribute matrices similar to \ODQR\ and for \HQRb\ we distribute matrices (2D-) block-cyclically with $b \times b$ blocks:
the distribution block size matches the algorithmic block size.
We parallelize \HQR\ and the base case of \HQRb\ to match the distribution of $\A$, analogous to \ODQR.
We parallelize \HQRb's inductive\-/case matrix multiplications to match the output matrix distribution (a 2D parallelization).
In the case of \HQR\ we assume $P = O(m)$.
In the case of \HQRb, 
we choose an $r \times c$ processor grid with $c = \Theta ( (nP/m)^{1/2} )$ and $r = \Theta(P / c)$,
and we choose $b = \Theta(1)$.
Assuming $P = \Omega(m \slash n)$ and
$P \cdot (\log P)^2 = O(m \cdot n)$,
these choices are valid and simultaneously minimize all three costs, asymptotically.

%

\CAQR\ \cite{DGHL12} modifies \HQRb\ to invoke \TSQR\ (\cref{sec:CAQR}) in the base case.
Our algorithms make crucial use of the \TSQR\ enhancements in \cite{BDGJKN15}; additionally, we use that paper's improved \CAQR\ in the following comparison. 
We parallelize and distribute data for \TSQR\ as discussed in \cref{sec:CAQR}, and for \CAQR's inductive case as we did for \HQRb's.
\TSQR\ and \CAQR's costs are summarized in the second rows of \cref{tab:costcomp1D,tab:costcomp2D}.
In the case of \TSQR\ we assume $P \le m/n$.
In the case of \CAQR\ we use the same $r \times c$ grid as for \HQRb\ but now pick $b = \Theta( n / (nP/m)^{1/2} )$.
Assuming $P \slash (\log P)^2 = \Omega(m \slash n)$ and
$P \cdot (\log P)^2 = O(m \cdot n)$,
these choices are valid and simultaneously minimize all three costs, asymptotically.

The costs of the new algorithms, \ODQR\ and \TDQR\, appear in the third rows of \cref{tab:costcomp1D,tab:costcomp2D}.
To make the comparison between \TSQR\ and \ODQR\ more clear, for the latter we use $b = \Theta(n / (\log P)^\epsilon)$ in \cref{thm:1DUB}'s proof, 
allowing the parameter $\epsilon$ to vary over $[0,1]$, 
justified by the stronger constraint $P (\log P)^{2\epsilon} = O\lt( n^2 \rt)$.
For \TDQR\ we follow \cref{thm:3DUB}'s proof and hypotheses.


\subsection{Elimination By Blocks}
\label{sec:tiskin}

%
%

Tiskin \cite{T07}, working in the BSP model \cite{V90}, 
proposed an algorithm outside of the \rQR\ framework which demonstrates a similar bandwidth\slash latency tradeoff as \TDQR.
Tiskin's algorithm was designed only for square matrices, but it has been extended to rectangular matrices, using the original algorithm as a black-box \cite{SBDH17}.
This extension achieves BSP bandwidth cost $O\lt(n^2/(nP/m)^\delta\rt)$ and BSP synchronization cost $O\lt((nP/m)^\delta (\log P)^2\rt)$.
Despite the fact that, in BSP, \TDQR\ achieves these same communication costs, we still believe \TDQR\ is a valuable contribution for multiple reasons. 
Defining a data distribution is a nontrivial and crucial step in developing a distributed-memory implementation using, e.g., MPI.
The aforementioned BSP algorithms do not (and need not) explicitly specify their data distributions.
Second, our algorithms are based on Householder's algorithm and use Householder representation. 
Thus, they are readily assembled from robust, tuned subroutines in standard libraries like (P)BLAS and (Sca)LAPACK.
Additionally, all interprocessor communication in our algorithms is expressed in terms of standard MPI collectives.
Lastly, we feel that Tiskin's recursive scheme, based on a slope-2 wavefront and `pseudopanels', is much more demanding from an implementation perspective than Elmroth-Gustavson's (\QREG).   
To our knowledge no one has implemented Tiskin's algorithm.

\subsection{Lower Bounds}
\label{sec:lbs}
Let us continue the notation from \cref{sec:compare}.
The algorithms studied there are all subject to an arithmetic lower bound of $\Omega( mn^2 / P)$ \cite{DGHL12}.

In the tall-skinny case, 
we have bandwidth and latency bounds $\Omega( n^2 )$ and $\Omega(\log P)$ \cite{CHPV07,BCDHKS14}.
\HQR\ attains the arithmetic lower bound,
but misses the bandwidth and latency lower bounds by $\Theta( \log P)$ and $\Theta( n )$.
\TSQR\ attains the arithmetic lower bound assuming $P \log P = \Omega (m / n)$, but 
misses the bandwidth and latency lower bounds both by $\Theta( \log P)$.
\ODQR\ attains the latency lower bound when $\epsilon = 0$, the arithmetic lower bound when $\epsilon \le 1/2$,
and the bandwidth lower bound when $\epsilon \ge 1$.

In the (close to) square case, 
we have bandwidth and latency bounds $\Omega( n^2 / (nP/m)^{2/3}$ and $\Omega( (nP/m)^{1/2})$ \cite{BCDHKS14}.
We restrict parameters so that both \CAQR\ and \TDQR\ attain the arithmetic lower bound, like \HQRb.
\HQRb\ and \CAQR\ exceed the bandwidth lower bound both by a factor of $\Theta( (nP/m)^{1/6})$ and 
they exceed the latency lower bound by factors of
$\Theta\lt( n / (nP/m)^{1/2} \log P\rt)$ and $\Theta\lt( (nP/m)^{1/6} (\log P)^2\rt)$, resp.
\TDQR\ attains the bandwidth lower bound when $\delta = 2/3$, and exceeds the latency lower bound by just $\Theta( (\log P)^2 )$ when $\delta = 1/3$.

We did not prove that \TDQR's bandwidth\-/latency product is optimal --- i.e., that the tradeoff is inevitable --- although we conjecture this to be the case.
Our intuition is based on bandwidth\slash latency tradeoffs observed in computations whose dependence graphs have similar diamond-shaped substructures: see, e.g.,
\cite{PU87,SCKD16}. 

\subsection{Limitations and Extensions}

Our main upper bound \cref{thm:3DUB} is substantially limited by its restrictions on permissible parallelism: see \cref{eq:3DUB-hyp}.
\TDQR's \AlltoAll{s} are responsible for these constraints: if we supposed the \AlltoAll{s} had zero cost, \cref{eq:3DUB-hyp} could be weakened to \cref{eq:3DQR-bstarconst}.
We make three remarks about improving this aspect of our work.

First, the bound used is worst-case; our knowledge of (and control over) data distribution could lead to stronger bounds. 
%
Second, it may be that the index algorithm is suboptimal for the data distribution, e.g., many $B_{pq}=0$ and a specialized algorithm would perform less communication, or at least yield sharper cost bounds.
Third and more generally, we should optimize for the data distribution before and after each subroutine.
%
%
The constraints are the balance assumptions to invoke \cref{lem:3DMM} and \ODQR. 
This is a difficult combinatorial problem.


The constants hidden in our asymptotic analysis are practically important, and the precise choices of parameters for particular machines warrants further study.
We have also omitted a number of practical optimizations that do not affect our asymptotic analysis.
For example, recall from \cref{sec:QR:basis-kernel} that $\T$ can be reconstructed from $\V$.
If the full $\T$ is not desired, by replacing the top level of recursion with a right\-/looking iterative \QREG\ variant, we can avoid ever computing superdiagonal blocks of $\T$; this does, however, restrict the available parallelism~\cite{EG00}.

\section*{Acknowledgments}

G. Ballard was supported by the National Science Foundation (NSF) Grant No. ACI-1642385.
The work of L. Grigori was supported by the NLAFET project as part of the European Union's Horizon 2020 research and innovation program under grant 671633.
Support for M. Jacquelin was provided in part through the Scientific Discovery through Advanced Computing (SciDAC) program funded by the US Department of Energy (DOE), Office of Science, Advanced Scientific Computing Research under Contract No. DE-AC02-05CH11231.

\bibliographystyle{alpha}

\newcommand{\etalchar}[1]{$^{#1}$}

\appendix

\section{Proof of Lemma 1}
\label{sec:collectives}

In \cref{sec:collectives} we give a more detailed proof of \cref{lem:collectives}.
First in \cref{sec:coll-mst} we review the \emph{binomial\-/tree algorithms} for \Scatter, \Gather, \Broadcast, \Reduce, and \AllReduce.
Then in \cref{sec:coll-bde} we review the \emph{bidirectional\-/exchange algorithms} for \ReduceScatter, \AllGather, \Broadcast, \Reduce, and \AllReduce.
Finally in \cref{sec:coll-ata} we review the \emph{(radix-2) index algorithm} for \AlltoAll, and a two-phase variant that admits sharper bounds when the block sizes vary widely.
For the four collectives with two algorithmic variants, we report the smaller of the two upper bounds in \cref{tab:collectives}.

We do not claim the chosen algorithms are optimal in our model or in practice.
When developing a high\-/performance implementation we would invoke the corresponding subroutines in a tuned MPI library.

\subsection{Binomial Tree Algorithms}
\label{sec:coll-mst}

The binomial\-/tree \Scatter\ algorithm proceeds as follows.
The algorithm terminates immediately when $P=1$; when $P >1$, we split the $P$ processors into two sets, of size $\lceil P/2 \rceil$ and $\lfloor P/2\rfloor$, and pick a processor~$r'$ in the set not containing processor~$r$.
Processor $r$ sends all blocks it owns destined for the opposite subset to processor~$r'$.
Processors~$r$ and~$r'$ then become the roots of two smaller \Scatter{s} among their respective processor subsets, which proceed in parallel.
An upper bound on the runtime satisfies
\[ T(P, B) \le T(\lceil P/2\rceil,\, B) + \beta\cdot \lceil P/2 \rceil B + \alpha\cdot 2\text, \]
which simplifies to 
\begin{equation}
\label{eq:ScatterUB}
\beta \cdot O\lt( (P-1)B\rt) + \alpha \cdot O\lt(\log P\rt)\text.
\end{equation}

The binomial tree \Broadcast\ algorithm is an optimization of the \Scatter\ algorithm when all blocks are identical: each message contains exactly one block, so the upper bound simplifies to 
\begin{equation}
\label{eq:BroadcastUB}
\beta \cdot O\lt( B\log P\rt) + \alpha \cdot O\lt(\log P\rt)\text.
\end{equation}

The binomial tree \Gather\ algorithm reverses the \Scatter's communication pattern, using head- (vs.\ tail-) recursion.
That is, after the recursive \Gather{s} with roots $r$ and $r'$ complete, processor~$r'$ sends processor~$r$ all blocks it owns.
An upper bound on the runtime satisfies the same recurrence as before, reproducing \cref{eq:ScatterUB}

The binomial tree \Reduce\ algorithm is an optimization of the \Gather\ algorithm when blocks are added as soon as they are received: each message contains exactly one block, so the upper bound simplifies to 
\begin{equation}
\label{eq:ReduceUB}
\gamma\cdot O\lt( B \log P\rt) + \beta \cdot O\lt( B\log P\rt) + \alpha \cdot O\lt(\log P\rt)\text.
\end{equation}

The binomial tree \AllReduce\ algorithm is a \Reduce\ followed by a \Broadcast, so its runtime also satisfies \cref{eq:ReduceUB}.

\subsection{Bidirectional Exchange Algorithms}
\label{sec:coll-bde}

The bidirectional exchange \ReduceScatter\ algorithm proceeds as follows.
The algorithm terminates immediately when $P=1$; when $P > 1$, we split the $P$ processors into two sets, of size $\lceil P/2 \rceil$ and $\lfloor P/2\rfloor$, and pair each processor with a processor in the other set.
Each processor sends its paired processor all blocks it owns destined for the other set.
If the sets differ in size, one processor~$p$ in the smaller set is paired with two processors $q,q'$ in the larger set:
processor~$p$ only sends to one of the two, but receives from both.
After each exchange, each processor adds the blocks it receives to the ones it already owns with the same destinations.
Each set then performs a smaller \ReduceScatter.
An upper bound on the runtime satisfies 
\begin{multline*} 
T(P, B) \le T(\lceil P/2\rceil,\, B) \\ + \gamma \cdot 2\lfloor P/2\rfloor + \beta \cdot\lt(\lceil P/2 \rceil + 2\lfloor P/2 \rfloor\rt) + \alpha \cdot 3\text,
\end{multline*}
which simplifies to
\[ \gamma\cdot O\lt( (P-1)B \rt) + \beta \cdot O\lt( (P-1)B\rt) + \alpha \cdot O\lt(\log P\rt) \text.\]

The bidirectional exchange \AllGather\ algorithm reverses the \ReduceScatter's communication pattern, using head- (vs.\ tail-) recursion.
That is, after the recursive \AllGather{s} complete, each processor sends its paired processor all blocks it owns destined for the other set, and if the sets differ in size, processor~$p$ only receives from one of $q,q'$, but sends to both.
Since blocks with the same sources are identical, each processor only sends only one copy of each. 
An upper bound on the runtime satisfies the same recurrence as before, except without the arithmetic cost, reproducing \cref{eq:ScatterUB}.

When $B$ is sufficiently large with respect to $P$, it is beneficial implement \Broadcast, \Reduce, and \AllReduce\ with the preceding bidirectional exchange algorithms, splitting the original blocks into new blocks of size at most $\lceil B/P\rceil$.
The bidirectional exchange \Broadcast\ algorithm is a \Scatter\ followed by an \AllGather, with runtime
\begin{equation}
\label{eq:BDEBroadcastUB}
\beta \cdot O\lt( (P-1)\lceil B/P \rceil\rt) + \alpha \cdot O\lt(\log P\rt) \text.
\end{equation}
The bidirectional exchange \Reduce\ algorithm is a \ReduceScatter\ followed by a \Gather, with runtime
\begin{multline}
\label{eq:BDEReduceUB}
\gamma\cdot O\lt( (P-1)\lceil B/P\rceil \rt) + \beta \cdot O\lt( (P-1)\lceil B/P \rceil\rt) \\+ \alpha \cdot O\lt(\log P\rt) 
\text.
\end{multline}
The bidirectional exchange \AllReduce\ algorithm is a \ReduceScatter\ followed by an \AllGather, so its runtime also satisfies \cref{eq:BDEReduceUB}.

\subsection{All-to-All}
\label{sec:coll-ata}

We consider an \AlltoAll\ algorithm called the \emph{radix-2 index algorithm} \cite{BHKUW97}.
Processors are numbered from $0$ and $P{-}1$, and each block is labeled $q{-}p \bmod P$, where $p$ is the source processor number and $q$ is the destination processor number.
Block labels are encoded in binary, as bit-strings of length $d = \lceil \log_2 P\rceil$, starting from their least significant bit. 
For each step $i = 0,1,\ldots,d{-}1$, each processor $p$ sends processor $p{+}2^i \bmod P$ a single message containing all blocks it currently owns whose labels' $i$-th bits are nonzero, at most $\lceil P/2 \rceil$ blocks. 
After $d$ steps, all blocks have arrived at their destination. 
The runtime is thus
\[ \beta \cdot O\lt( BP\log P \rt) + \alpha \cdot O\lt( \log P\rt)\text. \]

All \AlltoAll{}s in this work use a two-phase approach~\cite{HBJ96} that explicitly addresses variable block\-/sizes.
Each processor~$p$ starts by balancing its outgoing blocks, dealing the $B_{pq}$ elements of its original block destined for each processor~$q$ into new blocks destined for processors~$p{+}q$, $p{+}q{+}1$,~$\ldots$, and so on, cyclically.
The processors then perform two \AlltoAll{}s, the first to route the new blocks to intermediate processors, and the second to route elements to their original destinations. 
The largest block\-/sizes $B'$ and $B''$ in the first and second \AlltoAll{}s are bounded,
\begin{align*} 
B' &\le \frac{P-1}{2} + \max_q \sum_p \frac{B_{pq}}{P}
\\
B'' &\le \frac{P-1}{2} + \max_p \sum_q \frac{B_{pq}}{P}\text.
\end{align*}
The overall runtime is thus bounded by
\[ \beta \cdot O\lt(\lt(B_*+P^2\rt)\log P\rt) + \alpha\cdot O\lt(\log P\rt)\text,\]
where $B_* \le B P$ is the maximum number of words any processor holds before or after the collective,
\[ B_* = \max\lt(\max_q \sum_p B_{pq},\ \max_p \sum_q B_{pq}\rt)\text.\]

\section{Proof of Lemmas 2, 3, and 4}
\label{sec:MM}

Here we provide a more comprehensive analysis of matrix multiplication than in \cref{sec:3DMM}.

\subsection{Generic Algorithm}

Consider multiplying an $I \times K$ matrix $\A$ with a $K \times J$ matrix $\B$ to obtain an $I \times J$ matrix $\MC$, via the usual (entrywise) formula: for all $(i,j) \in [I]\times[J]$,
\[ \MC_{ij} = \sum_{k \in [K]} \A_{ik} \B_{kj} \text.\]

We identify each of the $IJK$ (scalar) multiplications with a point $(i,j,k)$ in three\-/dimensional Euclidean space, so the set of multiplications defines a discrete $I\times J \times K$ brick. 
We will parallelize the multiplications by arranging the processors in a three\-/dimensional grid and assigning each processor a sub-brick. 

A complication is that the admissible processor grids depend on the integer factorizations of the number $P$ of processors: for flexibility, we will allow some processors to remain idle.
That is, writing $P = QRS + T$ for any positive integers $Q$, $R$, and $S$ and nonnegative integer $T$, we arrange $QRS$ processors in a $Q \times R \times S$ grid, indexing each by a triple $(q,r,s)$, and set the remaining $T$ processors aside. 

Having fixed a processor grid, let
$\{\II_{q}\}_{q \in [Q]}$, 
$\{\JJ_{r}\}_{r \in [R]}$, and 
$\{\KK_{s}\}_{s \in [S]}$
denote partitions of $[I]$, $[J]$, and $[K]$.
Assign each grid processor $(q,r,s)$ the sub-brick $\II_{q}\times\JJ_{r}\times\KK_{s}$, meaning the (sub-) matrix product
\[ \M{Z}_{\II_q,\JJ_r,s} = \A_{\II_{q},\KK_{s}}\cdot\B_{\KK_{s},\JJ_{r}} \text. \]
Slices of the $I\times J \times S$ tensor $\M{Z}$ are partial sums of the output matrix: in particular, for all $(q,r) \in [Q]\times[R]$,
\[ \MC_{\II_{q},\JJ_{r}} = \sum_{s \in [S]} \M{Z}_{\II_q,\JJ_r,s} \text.\]
We will parallelize these remaining $I J(S{-}1)$ additions as part of data redistribution.

The initial distributions of $\A$ and $\B$, and the final distribution of $\MC$, are chosen as follows.
For each $(q,s) \in [Q]\times[S]$, $\A_{\II_q,\KK_s}$ is partitioned entrywise across the grid processors $(q,\cdot,s)$, and,
for each $(r,s) \in [R] \times[S]$, $\B_{\KK_s,\JJ_r}$ is partitioned entrywise across the grid processors $(\cdot,r,s)$.
We identify the matrix entries that grid processor $(q,r,s)$ owns with the sets 
$\cA_{q,r,s}\subseteq \II_q \times \KK_s$ 
and 
$\cB_{q,r,s} \subseteq \KK_s \times \JJ_r$.

During the computation, each of the matrices $\A$, $\B$, and $\MC$ undergoes a redistribution, using \AllGather{}s or \ReduceScatter{}s along fibers of the processor grid.

First,
for each $(q,s) \in [Q]\times[S]$, we perform an \AllGather\ among the grid processors $(q,\cdot,s)$ so that each obtains a copy of $\A_{\II_q,\KK_s}$;
these $QS$ \AllGather{}s, each involving a disjoint set of $R$ processors, occur simultaneously.
Next, 
for each $(r,s) \in [R]\times[S]$, we perform an \AllGather\ among the grid processors $(\cdot,r,s)$ so that each obtains a copy of $\B_{\KK_s,\JJ_r}$; similarly,
these $RS$ \AllGather{}s, each involving a disjoint set of $Q$ processors, occur simultaneously.

At this point, each grid processor $(q,r,s)$ can evaluate its local matrix product, obtaining the matrix $\M{Z}_{\II_q,\JJ_r,s}$ as explained above.

Finally, 
to construct the output matrix, 
for each $(q,r) \in [Q]\times[R]$, we perform a \ReduceScatter\ among the grid processors $(q,r,\cdot)$, computing $\MC_{\II_q,\JJ_r}$ while partitioning it entrywise across these processors;
these $QR$ \ReduceScatter{}s, each involving a disjoint set of $S$ processors, can occur simultaneously.
We identify the matrix entries that grid processor $(q,r,s)$ owns after these collectives with the set $\cC_{q,r,s}\subseteq \II_q \times \JJ_r$.

\subsection{Analysis}
\label{sec:3DMM:cost}

We denote by \MM\ the multiplication of an $I \times K$ matrix with a $K \times J$ matrix, both stored locally on one processor: this is the special case of \TDMM\ on a $1\times1\times1$ processor grid.
Using the conventional algorithm, which involves $IJK$ multiplications and $IJ(K{-}1)$ additions, \MM's runtime is bounded, \[ T_\MM\lt(I,J,K\rt) = \gamma\cdot O\lt(IJK\rt)\text.\]
(This was \cref{lem:0DMM}.)

The algorithmic costs of \TDMM\ are due to the \MM{s}, \AllGather{s}, and \ReduceScatter{s}. 
The costs of these subroutines, in turn, depend on the splitting $P=QRS+T$ and the six partitions,
\[\begin{array}{c c c c c c c}
\{\II_q\}_q & \text{of} & [I]\text, & \ & \{\cA_{q,r,s}\}_{q,r,s} & \text{of} & [I] \times [K]\text, \\
\{\JJ_r\}_r & \text{of} & [J]\text, & \ & \{\cB_{q,r,s}\}_{q,r,s} & \text{of} & [K] \times [J]\text, \\
\{\KK_s\}_s & \text{of} & [K]\text, & \ & \{\cC_{q,r,s}\}_{q,r,s} & \text{of} & [I] \times [J]\text.
\end{array}\]

\begin{lemma}
\label{lem:3DMM-gen}
\TDMM's runtime is bounded by
\begin{align*}
&\max_{q,r,s}\ T_\MM\lt(|\II_q|,\,|\JJ_r|,\,|\KK_s|\rt) \\+\;\;
&\max_{q,s}\ T_\AllGather\lt(\{\cA_{q,r,s}\}_r \rt) \\+\;\;
&\max_{r,s}\ T_\AllGather\lt(\{\cB_{q,r,s}\}_q \rt) \\+\;\;
&\max_{q,r}\ T_\ReduceScatter\lt(\{\cC_{q,r,s}\}_s \rt)
\text.
\end{align*}
\end{lemma}

Upper bounds on $T_\AllGather$ and $T_\ReduceScatter$ are given in \cref{lem:collectives}.
However, those results are pessimistic when the block-sizes vary.
For example, on a one\-/dimensional processor grid, an \AllGather\ or \ReduceScatter\ simplifies to a \Broadcast\ or \Reduce\ and sharper bounds apply.
This situation arises in \ODQR, in whose analysis we will additionally constrain the number of processors to ensure that the arithmetic and bandwidth costs of the \Broadcast\slash\Reduce\ are independent of $P$.
\cref{lem:1DMM} is a corollary of \cref{lem:3DMM-gen} that summarizes this special case.

\begin{proof}[of \cref{lem:1DMM}]
Suppose first that $K = \max(I,J,K)$.
Pick $Q=R=1$ and $S = P$, which is a valid $Q \times R \times S$ processor grid since $Q,R,S$ are positive integers with $QRS \le P$.
Fix any balanced partitions $\{\II_q\}_q$, $\{\JJ_r\}_r$, and $\{\KK_s\}_s$.
Take the partition $\{\cA_{q,r,s}\}_{q,r,s}$ (resp., $\{\cB_{q,r,s}\}_{q,r,s}$) to be a union of balanced $R$-way (resp., $Q$-way) partitions of the sets $\II_q \times \KK_s$ ($(q,s) \in [Q]\times[S]$), resp., $\KK_s \times \JJ_r$ ($(s,r) \in [S]\times[R]$).  
Finally, take $\cC_{q,r,s}=\emptyset$ for all but one $(q,r,s)$.
In this scenario, the \AllGather{s} are trivial (involving one processor each), and the single \ReduceScatter\ is just a \Reduce.
The arithmetic cost (\MM{s} and \Reduce) is
\[ 
\gamma\cdot O\lt( IJ\lt\lceil\frac{K}{P}\rt\rceil + \min(IJ\log P, IJ + P)\rt)
=
\gamma\cdot O\lt( \frac{IJK}{P} \rt)
\text.
\]
The bandwidth cost (\Reduce) is
\[
\beta\cdot O\lt(\min(IJ\log P,IJ+P)\rt) 
= 
\beta\cdot O\lt(IJ \rt)
\text.
\] 
The latency cost (\Reduce) is $\alpha\cdot O\lt(\log P\rt)$.

Now suppose that $I$ (resp., $J$)~$= \max(I,J,K)$ 
Pick $R=S=1$ and $Q=P$ (resp., $Q=S=1$ and $R=P$).
Again take any balanced partitions $\{\II_q\}_q$, $\{\JJ_r\}_r$, and $\{\KK_s\}_s$.
Similarly to before, take the partitions $\{\cA_{q,r,s}\}_{q,r,s}$ and $\{\cC_{q,r,s}\}_{q,r,s}$ (resp., $\cB_{q,r,s}$ and $\cC_{q,r,s}$)
to be unions of balanced $R$- (resp., $Q$-) and $S$-way partitions, 
while taking take $\cB_{q,r,s}$ (resp., $\cA_{q,r,s}$)~$=\emptyset$ for all but one $(q,r,s)$.
In this scenario, the \AllGather{s} of the left (resp., right) factor and the \ReduceScatter{s} are trivial (involving one processor each), and the single \AllGather\ of the right (resp., left) factor is just a \Broadcast.
The arithmetic cost is $\gamma \cdot O\lt( \lt\lceil I/P \rt\rceil JK\rt)$, resp., $\gamma\cdot O\lt( I\lt\lceil J/P\rt\rceil K\rt)$, which~$ = \gamma\cdot O\lt( IJK/P \rt)$.
Similar to the first case, the bandwidth cost is $\beta \cdot O\lt(JK\rt)$, resp.,  $\beta \cdot O\lt(IK\rt)$,
and the latency costs are both $\alpha \cdot O\lt( \log P\rt)$.
\end{proof}

\cref{lem:3DMM} is another corollary of \cref{lem:3DMM-gen}, applicable for matrix and processor grid dimensions
sufficiently large to admit parallelizations where each processor is assigned a roughly cubical sub-brick of the $I\times J \times K$ computation brick.
This situation arises in \TDQR.
\begin{proof}[of \cref{lem:3DMM}]
Pick 
$Q = \lfloor I / \rho \rfloor$, 
$R = \lfloor J / \rho \rfloor$, and
$S = \lfloor K / \rho \rfloor$, 
where
$\rho = (IJK/P)^{1/3}$.
Then
\begin{gather*}
\lt(1-1/c\rt)^3 P \le QRS \le P \\
1 \le \rho \le \frac{I}{Q},\frac{J}{R},\frac{K}{S} \le \frac{c}{c-1} \rho\text.
\end{gather*}
This defines a valid $Q \times R \times S$ processor grid, since $Q,R,S$ are positive integers with $QRS \le P$.
Take any balanced partitions $\{\II_q\}_q$, $\{\JJ_r\}_r$, and $\{\KK_s\}_s$.
Take the partition $\{\cA_{q,r,s}\}_{q,r,s}$ to be a union of balanced $R$-way partitions of the sets $\II_q \times \KK_s$ ($(q,s) \in [Q]\times[S]$), and similarly for the partitions $\{\cB_{q,r,s}\}_{q,r,s}$ and $\{\cC_{q,r,s}\}_{q,r,s}$.
The arithmetic cost (from \MM{s} and \ReduceScatter{s}) is
\begin{multline*}
\gamma\cdot O\lt(\lt\lceil\frac{I}{Q}\rt\rceil\lt\lceil\frac{J}{R}\rt\rceil\lt\lceil\frac{K}{S}\rt\rceil+(S-1)\lt\lceil\frac{\lt\lceil\frac{I}{Q}\rt\rceil\lt\lceil\frac{J}{R}\rt\rceil}{S}\rt\rceil\rt)
\\= \gamma\cdot O\lt( \frac{IJK}{P} \rt)
\text,
\end{multline*}
where a left-to-right reading of this equality introduces the assumption on $P$ from the theorem statement and fixes $Q,R,S$ as in the preceding paragraph.
The bandwidth cost (from \AllGather{s}\slash\ReduceScatter{s}) is
\begin{multline*}
\beta\cdot O\Bigg(
(R-1)\lt\lceil\frac{\lt\lceil\frac{I}{Q}\rt\rceil\lt\lceil\frac{K}{S}\rt\rceil}{R}\rt\rceil   +
(Q-1)\lt\lceil\frac{\lt\lceil\frac{J}{R}\rt\rceil\lt\lceil\frac{K}{S}\rt\rceil}{Q}\rt\rceil \\+
(S-1)\lt\lceil\frac{\lt\lceil\frac{I}{Q}\rt\rceil\lt\lceil\frac{J}{R}\rt\rceil}{S}\rt\rceil \Bigg)
 = \beta\cdot O\lt(\lt(\frac{IJK}{P}\rt)^{2/3}\rt)\text, 
\end{multline*}
where a left-to-right reading is as before.
The latency cost (from \AllGather{s}\slash\ReduceScatter{s}) is
\[ 
\alpha\cdot O\lt(\log R + \log Q + \log S \rt) 
=\alpha\cdot O\lt( \log P\rt)
\text.
\]
\end{proof}

\section{Proof of Lemma 5}
\label{sec:TSQR}

Here we provide relatively self-contained description of \TSQR.
To establish our asymptotic claims it suffices to use a simplified version of \TSQR, 
which we present in \cref{sec:CAQR-up,sec:CAQR-down} and analyze in \cref{sec:CAQR-an}.
Proofs of correctness and numerical stability can be found in~\cite{BDGJKN15}.

Recall that on input, the $m \times n$ matrix $\A$ is partitioned across the $P$ processors so that each processor~$p$ owns $m_p \ge n$ rows, not necessarily contiguous.
Thus we require $\A$ be sufficiently tall and skinny: $m/n \ge P$.
A single processor~$r$, which owns $\A$'s $n$ leading rows, is designated as the \emph{root processor}.

On output, the Q-factor is stored in Householder representation $(\V,\,\T)$, where $\V$ has the same distribution as $\A$.
Both $\T$ and the R-factor are returned only on the root processor.

Looking just at its communication pattern, \TSQR\ resembles a \Reduce\ followed by a \Broadcast, the distinction being the local arithmetic performed before and after each exchange.  

\subsection{Upsweep}
\label{sec:CAQR-up}

At the start of \TSQR\, each processor~$p$ performs a QR decomposition of $\A_p$, its $m_p \times n$ submatrix of $\A$.
\[ \lt(\V^{(0)}_p,\,\T^{(0)}_p,\,\R^{(0)}_p\rt) = \textsc{local-QR}\lt(\A_p\rt)\text.\]
The subroutine \textsc{local-QR} is unspecified other than that it computes a QR decomposition of an $\mu \times \nu$ matrix ($\mu \ge \nu$), stored locally, in $O(\mu\nu^2)$ operations, and returns the Q- and R\-/factors in the compact representations described in \cref{sec:QR:basis-kernel}.

Next, the processors perform a \Reduce\ using the binomial tree algorithm (\cref{sec:coll-mst}) with root~$r$ and blocks $\{\R_p\}_p$, meaning the block\-/size is $n(n+1)/2$. 
However, instead of adding the blocks elementwise after each exchange, we perform local QR decompositions:
\[ \lt(\V^{(\ell)}_p,\,\T^{(\ell)}_p,\,\R^{(\ell)}_p\rt) = \textsc{local-QR}\lt(\begin{bmatrix} \R^{(\ell-1)}_p \\\ \R^{(\ell-1)}_q\end{bmatrix}\rt)\text,\]
where $\R^{(\ell-1)}_p$ is processor~$p$'s R-factor from its previous QR decomposition and $\R^{(\ell-1)}_q$ is the R-factor it just received from some other processor~$q\ne p$.
Each processor~$p$ keeps the Q\-/factors (in basis-kernel representation) it produces --- they will be used subsequently --- but the R\-/factors are destroyed once they are sent on to another processor. 
At the end of the \Reduce, each processor stores between $1$ and $L = \lceil \log_2 P\rceil$ intermediate Q\-/factors, and processor~$r$ also stores $\R^{(L)}_r$, an R\-/factor of a QR decomposition of $\A$.

\subsection{Downsweep}
\label{sec:CAQR-down}

In principle, we can recover the Q-factor (in Householder representation) almost directly from $\A$ and the R-factor: see \cite[Section~4.4]{BDGJKN15} for a survey of approaches.
However, numerical issues motivate taking an additional step, recovering the leading $n$ columns of the Q-factor, which can be done stably and efficiently by applying the tree of Q-factors to a set of $n$ identity columns~\cite{DGHL12}.
This resembles a \Broadcast\ using the binomial tree algorithm (\cref{sec:coll-mst}) with root~$r$ and block\-/size $n^2$, reversing the communication pattern of the \Reduce.
Unlike a typical \Broadcast, however, the block's contents change each time it is sent to another processor.
That is, whenever processor~$p$ received a block from processor~$q$ during the \Reduce, processor~$p$ now computes, 
\[ \begin{bmatrix}\M{B}^{(\ell-1)}_p \\ \M{B}^{(\ell-1)}_q\end{bmatrix} = \lt(\M{I}-\V^{(\ell)}_p \T^{(\ell)}_p (\V^{(\ell)}_p)^H\rt)\begin{bmatrix}\M{B}^{(\ell)}_p\\\M{0}\end{bmatrix}\]
and then sends $\M{B}^{(\ell)}_q$ to processor~$q$.
To start, the root processor~$r$ sets $\M{B}^{(L)}_r = \M{I}$.

Next, each processor~$p$ computes 
\[ \M{W}_p = \lt(\M{I}-\V^{(0)}_p \T^{(0)}_p (\V^{(0)}_p)^H\rt)\lt[\begin{smallmatrix}\M{B}^{(0)}_p\\\M{0}\end{smallmatrix}\rt]\text.\]
The matrix $\M{W}$, defined by the submatrices $\{\M{W}_p\}_p$ just as $\A$ is defined by $\{ \A_p \}_p$, is the rightmost $m\times n$ submatrix of the Q-factor associated with the R-factor obtained by the \Reduce.

It remains to recover a Householder representation of the Q-factor from $\M{W}$.
For numerical stability, we exploit the non-uniqueness of a QR decomposition $(\Q,\R)$ of $\A$: 
$(\Q\M{Z},\M{Z}^H \R)$ is also a QR decomposition of $\A$ for any unitary matrix $\M{Z}$ with $2 \times 2$ block-diagonal structure whose leading block is an $n \times n$ diagonal matrix.
Processor~$r$ row-reduces $\M{X}$, the upper $n\times n$ submatrix of $\M{W}_r$, in the usual approach ---
working left-to-right, eliminating each column below the diagonal by premultiplying with a unit lower triangular matrix ---
but with a modification to simultaneously populate an $n \times n$ diagonal matrix $\M{S}$.
For each $j \in [n]$, before the $j$-th column is eliminated, letting $\hat{\M{X}}$ denote the current partially reduced matrix, compute $\M{S}_{jj}=\mathrm{sgn}(\hat{\M{X}}_{jj})$, add $\M{S}_{jj}$ to $\hat{\M{X}}_{jj}$, and then proceed as usual.
(Here $\mathrm{sgn}(z)$ means $z/|z|$ for $z \ne 0$ and an arbitrary unit complex number when $z=0$.)
Since $\hat{\M{X}}_{jj}+\M{S}_{jj} \ne 0$, pivoting is unnecessary to avoid breakdown, thus this procedure terminates with both $\M{S}$ and a matrix pair $(\M{L},\M{U})$, where $\M{L}$ is unit lower triangular, $U$ is upper triangular, and $\M{L}\M{U}$ is invertible.
It is shown in~\cite[Lemma~6.2]{BDGJKN15} that $\M{X}+\M{S}=\M{L}\M{U}$ and, moreover, that partial pivoting is obviated (or performed implicitly): at each step, the diagonal element's magnitude is at least that of each element in the column below.
Processor~$r$ then computes 
$\T = \M{U}\M{S}^H\M{L}^{-H}$, 
$\V_r = \M{L}$, and 
$\R = -\M{S}^H \R^{(L)}_r$.
The processors then perform a \Broadcast\ of $\M{U}$, with root~$r$, after which each recipient processor~$p \ne r$ computes $\V_p = \M{W}_p\M{U}^{-1}$.

We are done: the Q-factor's basis $\V$ is partitioned as $\{\V_p\}_p$ across the processors commensurately with $\A$, and the root processor stores both the kernel $\T$ and the R-factor $\R$. 

\subsection{Concluding the Analysis}
\label{sec:CAQR-an}

\begin{proof}[of \cref{lem:TSQR}]
The runtime of the \textsc{local-QR}s, performed in parallel, is $\gamma \cdot O\lt(\max_p m_p n^2 \rt)$, and 
the runtime of the subsequent `\Reduce'\ is
\[ \gamma\cdot O\lt( n^3 \log P\rt) + \beta\cdot O\lt( n^2 \log P \rt) + \alpha\cdot O\lt( \log P\rt)\text.\]
Computing $\M{W}$ asymptotically matches the current total runtime.
Computing $\M{U}$ on the root processor adds $\gamma \cdot O\lt( n^3 \rt)$, then 
the subsequent \Broadcast\ and computing $\M{V}$ matches computing $\M{W}$.
\end{proof}

Note that in \cref{sec:CAQR-down} it is possible to avoid the second \Broadcast\ and computing most of $\M{W}$ by starting the first \Broadcast\ with $\M{B}^{(L)}_r = \M{U}^{-1}$ \cite{BDGJKN15}.

\end{document}